\documentclass[twoside,leqno]{article}
\usepackage[letterpaper]{geometry}
\usepackage{ltexpprt}
\usepackage{hyperref}
\usepackage[utf8]{inputenc}
\usepackage{graphicx}
\usepackage{float}
\usepackage{booktabs}
\usepackage{array}
\usepackage{xfrac}
\usepackage{subcaption}
\usepackage{parskip}
\usepackage[english]{babel}
\usepackage{tikz}
\usepackage{amssymb}
\usepackage{amsmath}
\usepackage{fancyhdr}
\usepackage{wrapfig}
\usepackage{mathtools}
\usepackage{mdframed}
\usepackage{nicefrac}
\usepackage{esvect}
\usepackage{hyperref}
\usepackage{csquotes}
\usepackage[disable]{todonotes}
\usepackage{authblk}

\usetikzlibrary{calc}
\usetikzlibrary{arrows}
\usetikzlibrary{arrows.meta,positioning}

\usepackage{pxfonts}
\usepackage[T1]{fontenc}
\usepackage{mathrsfs}

\newcommand{\T}{\ensuremath{\mathcal T}}

\DeclareMathOperator{\depth}{depth}
\DeclareMathOperator{\sibling}{sibling}
\DeclareMathOperator{\treeroot}{root}

\DeclarePairedDelimiter{\ceil}{\lceil}{\rceil}
\DeclarePairedDelimiter{\floor}{\lfloor}{\rfloor}
\DeclarePairedDelimiter{\paren}{\lparen}{\rparen}

\title{Splay Top Trees\relatedversion{}}
\date{}

\author[1]{Jacob Holm\thanks{Partially supported by the VILLUM Foundation grant 16582, ``BARC".}}
\affil[1]{Department of Computer Science, University of Copenhagen, DK-2100 Copenhagen, Denmark}

\author[2]{Eva Rotenberg\thanks{Partially supported by Independent Research Fund Denmark grant 2020-2023 (9131-00044B) ``Dynamic Network Analysis'', and the VILLUM Foundation grant 37507 ``Efficient Recomputations for Changeful Problems''.}}
\affil[2]{Department of Applied Mathematics and Computer Science, Technical University of Denmark,
DK-2800 Lyngby, Denmark}

\author[2]{Alice Ryhl\thanks{Partially supported by the VILLUM Foundation grant
37507 ``Efficient Recomputations for Changeful Problems''. Also a Google
employee.}}

\usepackage{lastpage}
\begin{document}

\newcommand\relatedversion{}

\maketitle





\fancyfoot[R]{\scriptsize{Copyright \textcopyright\ 2023\\
This paper is available under \href{https://creativecommons.org/licenses/by/4.0/}{the CC-BY 4.0 license}}}


\vspace*{-0.6cm}
\begin{abstract}
The \emph{top tree} data structure is an important and fundamental tool in dynamic graph algorithms. Top trees have existed for decades, and today serve as an ingredient in many state-of-the-art algorithms for dynamic graphs. In this work, we give a new direct proof of the existence of top trees, facilitating simpler and more direct implementations of top trees, based on ideas from splay trees. This result hinges on new insights into the structure of top trees, and in particular the structure of each root path in a top tree.
In amortized analysis, our top trees match the asymptotic bounds of the state of the art. 
\end{abstract}


%
%
%
%
\pagenumbering{arabic}


\section{Introduction}

An interesting topic in graph algorithms is \emph{dynamic graphs}; graphs that
undergo updates or changes.
Throughout the past half century, fundamental algorithmic questions about graphs
such as connectivity, minimum spanning tree, shortest paths, matchings, and
planarity, have been studied in the dynamic setting, where edges may appear and
disappear adversarially \cite{
Frederickson85,Eppstein97,HolmLT01,Thorup00,Kapron:2013,Wulff-Nilsen16a,NSW17,HolmRW15,patrascu06,%
DemetrescuI04,BrandN19, GutenbergW20, AbboudD16, KarczmarzMS22,%
Sankowski07, OnakR10,Solomon16,BaswanaGS18,BhattacharyaK19,BernsteinFH21,GrandoniSSU22,%
GalilIS99,EppsteinGIS96,Eppstein03,HolmR20%
}. For such problems, researchers engage in the pursuit of
a data structure with an (as efficient as possible) polylogarithmic algorithm for
updates and queries, or (conditional) lower bounds certifying that such
algorithms are unlikely to exist. For a wide range of problems, polylogarithmic
time data structures do indeed exist, and often, such algorithms use \emph{top
trees} as a tool towards efficient solutions.

The top tree~\cite{alstrup2005maintaining} is a data structure, which can be
used in combination with any (explicitly maintained) dynamic spanning forest of
the changing graph. The strength of top trees is that they can be used to
efficiently store concise summaries of information about subtrees of a tree in a
way that allows for efficient updates and queries. If we think of the dynamic graph problem
as a calculation, the top tree stores a tree of summaries of properties of subgraphs, and upon changes to the graph, we only need to recompute a logarithmically bounded number
of subgraph summaries.

More formally, the basic structure of a top tree can be seen as a binary tree of
contractions in the underlying tree. The leaf nodes of the top tree are the
edges in the underlying tree, and the internal nodes correspond to the union, or
\emph{merge}, of the two children. Each node must behave as a ``generalized
edge'' in the sense that it may share at most two vertices with the rest of the
graph. We denote by \emph{cluster} the set of edges corresponding to a node in
the top tree. Top trees build on similar principles as topology trees, which is
the tool introduced and used by Frederickson~\cite{Frederickson85} to get the
first non-trivial update algorithms for connectivity-related questions about
dynamic graphs~\cite{Frederickson85,Frederickson91,ItalianoPR93}. However, top
trees and topology trees differ in the fundamental way that while top trees use
a notion of clusters that generalizes edges, topology trees use a notion of
clusters that generalizes low-degree vertices.

Not only do top trees build on similar principles as topology trees, the
original proof that top trees of logarithmic height can be maintained goes via
reduction to topology trees~\cite[Section 6.2]{alstrup2005maintaining}. Most
implementations of top trees described to date have the form of either an extra
interface imposed on topology trees, or are implemented as an interface on
Sleator and Tarjan's link-cut trees~\cite{TarjanWerneck05}. 

In this work, we show a different route to top trees, and provide new proofs of top tree properties. We show that top trees with logarithmic amortized update times may be directly implemented in a manner akin to splay trees. We claim that such direct implementations are simpler, and provide pseudo-code and C code as proof of concept. While our approach is indeed simple to implement, it is not trivial that it has good theoretical guarantees, and we give proofs of its efficient amortized update times.

Just like the state-of-the art implementation by Tarjan and Werneck~\cite{TarjanWerneck05}, a caveat of our work is that all times are only amortized. However there are many results involving top trees for which this is not a problem. For one, there are many dynamic problems for which the efficient (polylog-time) solution using top trees is already amortized. Examples include  biconnectivity, $2$-edge connectivity, planarity testing, and diameter~\cite{HolmLT01,Thorup00, HolmRT17,HolmR17,HolmR20,alstrup2005maintaining}. Secondly, dynamic algorithms have a range of applications in static problems, where the solution to the static problem is computed via a sequence of subproblems, dynamically extracted from the original problem~\cite{BiedlBDL01,GabowKT01,MuchaS04,BarrB21}. When used as a subroutine this way, amortized running time is just as valuable as worst-case.

\textbf{Related work.}
Data structures for dynamic forests include link-cut trees that were introduced
by Sleator and Tarjan~\cite{SleatorTarjan83} and who used it, among other applications,
to solve flow-problems. Topology trees were introduced by
Frederickson~\cite{Frederickson85} and used to obtain the first sublinear-time
algorithms for dynamic graph connectivity. Later, top
trees~\cite{alstrup2005maintaining} were introduced and used to give new
polylogarithmic update-time algorithms for connectivity and related
problems~\cite{HolmLT01,Holm18a}.

Splay trees were introduced by Sleator and Tarjan~\cite{SleatorTarjan85} under
the name of \emph{Self-adjusting Binary Search trees}, as a simple way to obtain
efficient amortized time per operation, without guaranteeing that the tree is
always balanced.
The same paper also introduced the notion of \emph{semi-splaying}
as an attempt to reduce the total number of rotations
by a constant, while still guaranteeing efficient amortized time per operation. In~\cite{BrinkmannDegraerLoof09} it was confirmed empirically that using semi-splaying, ``the practical performance is better for a very broad variety of access patterns''.

Splay trees have been conjectured to be \emph{dynamically
  optimal}~\cite{ColeMSS00,Cole00}, meaning that for any sequence of operations, they do at most a constant factor more work than an optimal algorithm that knows the whole sequence of operations in advance. The subject of dynamic
optimality has been addressed also in connection with Tango
trees~\cite{DemaineHIP07}, and in recent advances on search trees on
trees~\cite{BoseCIKL20,BerendsohnKozma22}. 


%

In practice, the state-of-the-art implementation of top trees is by Tarjan and Werneck~\cite{TarjanWerneck05}. Their implementation is based on link-cut trees, that are enhanced with extra information, and their update-times are amortized. An experimental evaluation of topology trees in comparison with link-cut-trees is given by Frederickson~\cite{Frederickson97}, and Albers and Karpinski~\cite{AlbersKarpinski02} study splay trees in theory and practice, proposing a randomized variant with fast (i.e. $O(\log n)$) expected worst-case time per operation.



\textbf{Overview of paper.}
First, in Section~\ref{sec:overview-ds}, we give an overview of the data
structure and the operations we support. Then,
sections~\ref{sec:has-x-boundary}, \ref{sec:rotate}, \ref{sec:splay},
\ref{sec:consuming}, and \ref{sec:dynamic-operations} are dedicated to the
details of those data structures and operations:
In Section~\ref{sec:has-x-boundary}, we describe some very simple local queries in the top tree.
In Section~\ref{sec:rotate}, we describe the rotation-like subroutine that the splay tree uses.
Section~\ref{sec:splay} is devoted to splaying, and we describe and prove structural properties that make both semi-splaying and splaying easy to implement in $O(\log n)$ amortized time.
Section~\ref{sec:consuming} describes an important internal operation used as a subroutine by the top tree for several of its fundamental operations.
Section~\ref{sec:dynamic-operations} describes the dynamic operations that are called directly by the user of the top tree: \texttt{expose}, \texttt{deexpose}, \texttt{link}, and \texttt{cut}.
In Section~\ref{sec:implement}, we provide practical information about our implementation, which is available for download.
In Appendix~\ref{sec:expose-semi}, we describe a more complicated way to implement \texttt{expose} that does not rely on \texttt{full\_splay}. This reduces the number of changes to the top tree, and we conjecture that it is faster in practice.

\section{Overview of the data structure}

\label{sec:overview-ds}
A \emph{top tree} is a data structure that represents an \emph{underlying tree},
usually from some underlying forest.  A top tree can be augmented to store additional information, e.g.\ edge weights, and to answer various types of queries on the underlying tree.  It does this by maintaining a hierarchy of \emph{summaries} of certain connected subgraphs of the underlying tree, from which the answers can be efficiently computed.
Up to two vertices in each underlying tree can be marked as
\emph{exposed}, which affects the structure of the top tree and what
summaries are stored, and therefore which queries the top tree is
ready to answer quickly at any given time.  With no exposed vertices this could e.g.\ be the size or diameter of the underlying tree. With one vertex exposed, we can think of the underlying tree as rooted in that vertex and we might e.g.\ answer queries about the height.  With two vertices exposed,
that would typically be some question related to the path between
them, e.g.\ the length or the maximum edge
weight. See~\cite{alstrup2005maintaining} for more details.

\subsection{External operations}\todo{Public API?}
Using \(T_v\) to refer to the tree in the underlying forest containing the
vertex \(v\), and \(\T_v\) to refer to the corresponding top tree, the
operations we support are as follows:
\begin{description}
\item[expose\((v)\):] Makes \(v\) exposed, and returns the root\footnote{If \(v\) is an isolated vertex, we define both \(\T_v\) and \(\operatorname{root}(\T_v)\) to be null.} of \(\T_v\).
  Requires that the tree \(T_v\) containing \(v\) has at most \(1\) exposed
  vertex and that \(v\) is not currently exposed. Does not change which node
  is the root node of \(\T_v\).

\item[deexpose\((v)\):] Makes \(v\) not exposed, and returns the root of
  \(\T_v\). Requires that \(v\) is currently exposed in the tree \(T_v\)
  containing it. Does not change which node is the root node of \(\T_v\).

\item[link\((u,v)\):] Creates a new edge \(u,v\) in the forest and
  returns the root of the top tree for the resulting tree. Requires
  that \(u,v\) are in disjoint trees \(T_u,T_v\) and that neither
  tree has any exposed vertices.

\item[cut\((e)\):] Deletes the edge \(e\) with endpoints \(u,v\) from
  the forest and returns\footnote{Our C code does \emph{not} return
    the roots of the resulting trees, because returning multiple
    values is inconvenient in C.} the roots of the two resulting trees
  \(\T_u,\T_v\). Requires that the tree containing \(e\) has no
  exposed vertices.
\end{description}
This differs slightly from the standard description of top trees where
\texttt{expose} is usually defined to take up to two arguments, but it should be clear
that it is equivalent in terms of power. We give a direct proof that the
amortized cost of each operation is \(O(\log n)\).

\subsection{Structure}
The top tree is itself a rooted binary\footnote{The definition of top
  trees makes perfect sense for higher degree nodes, but for
  simplicity we will restrict ourselves to the binary case.} tree of
nodes, where each node corresponds to a \emph{cluster}, which is a
connected set of edges in the underlying tree. Each leaf in the top
tree corresponds to a single edge in the underlying tree, and each
internal node corresponds to the disjoint union of its two children.

We say that a vertex is a \emph{boundary vertex} of a cluster if it is
incident to an edge in the cluster and either is exposed or is also
incident to an edge outside the cluster. A cluster is \emph{valid} if
it has at most two boundary vertices.  Conceptually, the
boundary vertices lie on the boundary between the cluster and the rest
of the tree, and since valid clusters have at most two boundary vertices,
valid clusters can be thought of as ``generalized edges'' with the endpoints
being the boundary vertices.  It is possible for such a generalized edge to have less than two endpoints, but we can think of them as normal edges whose ``missing'' endpoints are irrelevant to us.

The top tree must satisfy the invariant that all clusters are
valid. We note that this invariant requires that each tree in the
underlying forest has at most two exposed vertices, because the root
of the top tree would otherwise have too many boundary vertices.

We categorize the (valid) clusters in the top tree as either \emph{path
  clusters} or \emph{point clusters}. A path cluster is a cluster with
two boundary vertices, and a point cluster is a cluster with zero or
one boundary vertices. For each internal node in the top tree, we also
define its \emph{central vertex} as the vertex shared by its two
children. The set of boundary vertices in a node is always equal to
the union of boundary vertices in its children, with or without the
central vertex removed.

\subsection{Orientation invariant}

When working with top trees, it is useful to think of the children of each node
to have a specific left-to-right order, or \emph{orientation}. (I.e.\ one child
is the left child and the other child is the right child.)

Given the orientation of each node, we can define whether a boundary
vertex of a node is to the left, in the middle, or to the right. For
internal nodes, a boundary vertex is in the \emph{middle} if it is the
central vertex of the node. Otherwise, a boundary vertex is to the
\emph{right} if it comes from the right child, and to the \emph{left} if it comes
from the left child. The \emph{leftmost boundary vertex} is then defined as
the left boundary vertex if one exists, otherwise the middle boundary vertex if one
exists. If the node only has a right boundary vertex, then we say that
it has no leftmost boundary vertex. The definition of rightmost is
analogous.\todo{JH: How about calling them \emph{left-or-middle} boundary vertices instead? Easy to understand, and clearly not a right boundary vertex.}

For leaf nodes, the left boundary vertex is the left endpoint and the right
boundary vertex is the right endpoint, if those vertices are boundary vertices
of the edge. Leaf nodes never have a central or middle vertex.

Having \emph{some} convention for what orientation is chosen for each node helps reduce the number of cases that need to be considered, both for the operations explored in this paper, and for any augmentations that need to maintain additional information as part of the stored summaries.  In this paper we require the orientations in the top tree to satisfy the following
\emph{orientation invariant}: For any internal node, the leftmost
boundary vertex of the right child and the rightmost boundary vertex of the left
child must both exist and be equal to the central vertex of the node.

The invariant determines the orientation of every internal non-root node relative to the orientation of its parent, except for point clusters whose children are both point clusters.

A larger example of a top tree that satisfies the orientation invariant can be
found in Figure~\ref{fig:rootexample}. Note that \(b_5\) has its only boundary
vertex in the middle, so the subtree rooted at \(b_5\) could be mirrored without
breaking the invariant.
To maintain the orientation invariant as the structure changes, we sometimes
have to \emph{flip} the orientation of a whole subtree.

\begin{figure}
	\centering
	\begin{subfigure}{0.5\textwidth}
		\centering
		\begin{tikzpicture}[xscale=0.8, yscale=0.625, every node/.style={circle,inner sep=1pt}]
		\begin{scope}
		\node (c0) at (1,0) {\(c_0\)};
		\node (c1) at (0,0) {\(c_1\)};
		\node (c2) at (1.5,{1*sqrt(0.75)}) {\(c_2\)};
		\node (c3) at (0,{2*sqrt(0.75)}) {\(c_3\)};
		\node (c4) at (1.5,{3*sqrt(0.75)}) {\(c_4\)};
		\node (c5) at (2,{4*sqrt(0.75)}) {\(c_5\)};
		\node (c6) at (0.5,{5*sqrt(0.75)}) {\(c_6\)};
		\node (b1) at (0.5,{1*sqrt(0.75)}) {\(b_1\)};
		\node (b2) at (1,{2*sqrt(0.75)}) {\(b_2\)};
		\node (b3) at (0.5,{3*sqrt(0.75)}) {\(b_3\)};
		\node (b4) at (1,{4*sqrt(0.75)}) {\(b_4\)};
		\node (b5) at (1.5,{5*sqrt(0.75)}) {\(b_5\)};
		\node (b6) at (1,{6*sqrt(0.75)}) {\(b_6\)};
		\end{scope}
		
		\begin{scope}[every edge/.style={draw=black,very thick},]
		\path (c0) edge (b1);
		\path (c1) edge (b1);
		\path (c2) edge (b2);
		\path (c3) edge (b3);
		\path (c4) edge (b4);
		\path (c5) edge (b5);
		\path (c6) edge (b6);
		\path (b1) edge (b2);
		\path (b2) edge (b3);
		\path (b3) edge (b4);
		\path (b4) edge (b5);
		\path (b5) edge (b6);
		\end{scope}
		\end{tikzpicture}
		\caption{The nodes as they appear in the top tree.}
	\end{subfigure}%
	\begin{subfigure}{0.5\textwidth}
		\centering
		\begin{tikzpicture}
		\begin{scope}[every node/.style={scale=0.6}]
		\node[circle,draw] (v1) at (0,0) {};
		\node[circle,draw] (v2) at (1,0) {};
		\node[circle,draw] (v3) at (1,1) {};
		\node[circle,draw] (v4) at (2,0) {};
		\node[circle,draw] (v5) at (1,2) {};
		\node[circle,draw] (v6) at (3,0) {};
		\node[circle,draw] (v7) at (3,1) {};
		\node[circle,draw] (v8) at (4,0) {};
		\end{scope}
		\begin{scope}[every edge/.style={draw,thick},]
		\path (v1) edge node[below] {\(c_0\)} (v2);
		\path (v2) edge node[left] {\(c_1\)} (v3);
		\path (v2) edge node[below] {\(c_2\)} (v4);
		\path (v3) edge node[left] {\(c_3\)} (v5);
		\path (v4) edge node[below] {\(c_4\)} (v6);
		\path (v6) edge node[left] {\(c_5\)} (v7);
		\path (v6) edge node[below] {\(c_6\)} (v8);
		\end{scope}
		\end{tikzpicture}
		\caption{The underlying tree.}
	\end{subfigure}
  \caption{An example of a top tree that satisfies the orientation invariant.}
	\label{fig:rootexample}
\end{figure}
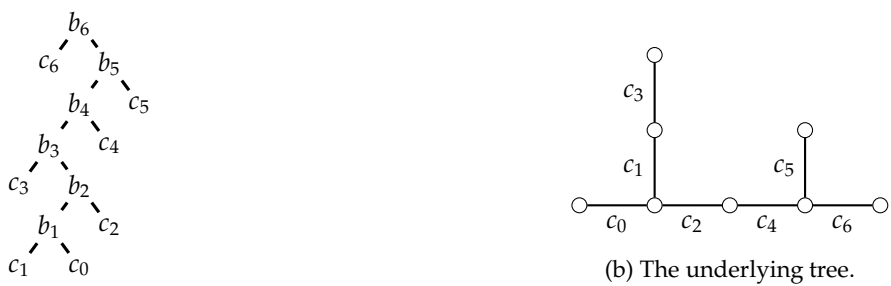

\subsection{Internal operations}

To support the external operations (\texttt{link}, \texttt{cut},
\texttt{expose}, \texttt{deexpose}) we introduce a collection of \emph{internal
operations}. Some of these are probably not directly useful for applications
using the top tree, but form the basis for implementing the external operations.
\begin{description}
\item[is\_point\((\mathtt{node})\), is\_path\((\mathtt{node})\)]
  Returns whether \texttt{node} is a point cluster or a path cluster.

\item[flip\((\mathtt{node})\)] To maintain the orientation
  invariant as the structure changes, we sometimes have to \emph{flip}
  the orientation of a whole subtree.  Flipping can be done by reversing the
  left-to-right order in every node of the subtree, but for efficiency we need to do this in a lazy
  fashion. The idea is to just store a boolean \emph{flip bit} in each
  node, indicating whether or not to conceptually flip the whole
  subtree rooted at that node.  This operation therefore just inverts
  that bit, which trivially takes constant time.  Since it is literally just
  that, our pseudocode manipulates that bit directly rather than calling a
  function to do it.

\item[push\_flip\((\mathtt{node})\)] This operation ensures that the
  flip bit of \texttt{node} is \texttt{false}, by actually swapping and flipping
  the children of \texttt{node} if necessary.  This also takes
  constant time.

  Since all our logic for the rest of the internal operations is
  symmetric, calling \texttt{push\_flip} on the constant number of
  nodes where the relative orientation is relevant is sufficient to
  ensure the part of the top tree we are working on is stored in a way
  that we can ignore the flip bits in much of the rest of the logic.

\item[has\_left\_boundary\((\mathtt{node})\),
  has\_middle\_boundary\((\mathtt{node})\),
  has\_right\_boundary\((\mathtt{node})\)]\leavevmode\\ Returns whether
  \texttt{node} has a left, middle, or right boundary vertex
  respectively. These functions can be implemented in constant time if
  we know the number of boundary vertices of each cluster. In
  practice, it is useful to ignore the flip bits of the proper
  ancestors of \texttt{node}, but not the flip bit on \texttt{node}
  itself when defining what left and right means for these functions.
  In other words, the result is as if \texttt{push\_flip} had been
  called on \texttt{node} first, so it is consistent with the
  orientation of the parent.

\item[rotate\_up\((\mathtt{node})\)] This is similar to the
  \texttt{rotate} operation known from binary search trees.  See
  Figure~\ref{fig:rotate-up} for an illustration.

  There are two major differences between \texttt{rotate\_up} and rotations in
  binary search trees. One is that rotations are not always allowed. The
  \texttt{rotate\_up} function is only allowed if \(\mathtt{sibling(node)}
  \cup \mathtt{sibling(parent(node))}\) is a valid cluster. The other
  difference is that it does not always preserve the ordering of
  the leaves.

  The operation can be defined in the following way: Swap the parent
  pointers of \texttt{node} and \texttt{sibling(parent(node))}, then adjust
  child pointers, orientations, and any other remaining fields in a way such
  that all invariants are satisfied. It can be shown that this is always
  possible if \(\mathtt{sibling(node)} \cup \mathtt{sibling(parent(node))}\)
  is a valid cluster.


\begin{figure}[h]
\begin{subfigure}{0.20\textwidth}
	\begin{tikzpicture}[scale=0.75]
		\begin{scope}
			\node (c0) at (2,0) {$A$};
			\node (c1) at (4,0) {$C$};
			\node (c2) at (3,{1*sqrt(3)}) {$A\cup C$};
		\node (b1) at (1,{1*sqrt(3)}) {$B$};
	\node (b2) at (2,{2*sqrt(3)}) {$A\cup B\cup C$};
\end{scope}

\begin{scope}[every edge/.style={draw=black,very thick},]
	\path (c0) edge (c2);
	\path (c1) edge (c2);
	\path (c2) edge (b2);
	\path (b1) edge (b2);
\end{scope}
\end{tikzpicture}
\end{subfigure}
\begin{subfigure}{0.16\textwidth}
  \begin{tikzpicture}[
      scale=0.75,
    ]
    \useasboundingbox (0,-3) rectangle (3,2);
    \draw[draw=purple!80!brown,thick,->] (3,0) to[bend right]  node[above,text=purple!80!brown] {\(\operatorname{rotate\_up}(B)\)} (0,0);
  \end{tikzpicture}
\end{subfigure}
	\begin{subfigure}{0.21\textwidth}
\begin{tikzpicture}[scale=0.75]
	\begin{scope}
		\node (c0) at (2,0) {$B$};
		\node (c1) at (0,0) {$A$};
		\node (c2) at (3,{1*sqrt(3)}) {$C$};
		\node (b1) at (1,{1*sqrt(3)}) {$A\cup B$};
		\node (b2) at (2,{2*sqrt(3)}) {$A\cup B\cup C$};
	\end{scope}
	
	\begin{scope}[every edge/.style={draw=black,very thick},]
		\path (c0) edge (b1);
		\path (c1) edge (b1);
		\path (c2) edge (b2);
		\path (b1) edge (b2);
	\end{scope}
\end{tikzpicture}
\end{subfigure}
\begin{subfigure}{0.18\textwidth}
  \begin{tikzpicture}[
      scale=0.75,
    ]
    \useasboundingbox (0,-3) rectangle (3,2);
    \draw[draw=cyan!70!black,thick,->] (0,0) to[bend left]  node[above,text=cyan!70!black] {\(\operatorname{rotate\_up}(A)\)} (3,0);
  \end{tikzpicture}
\end{subfigure}
\begin{subfigure}{0.23\textwidth}
\begin{tikzpicture}[scale=0.75]
	\begin{scope}
		\node (c0) at (2,0) {$B$};
		\node (c1) at (4,0) {$C$};
		\node (c2) at (3,{1*sqrt(3)}) {$B\cup C$};
		\node (b1) at (1,{1*sqrt(3)}) {$A$};
		\node (b2) at (2,{2*sqrt(3)}) {$A\cup B\cup C$};
	\end{scope}
	
	\begin{scope}[every edge/.style={draw=black,very thick},]
		\path (c0) edge (c2);
		\path (c1) edge (c2);
		\path (c2) edge (b2);
		\path (b1) edge (b2);
	\end{scope}
\end{tikzpicture}
\end{subfigure}
\caption{\label{fig:rotate-up}A top tree containing clusters $A$, $B$, and $C$ (middle), and the result of calling either \textcolor{purple!80!brown}{$\operatorname{rotate\_up}(B)$} (left) or \textcolor{cyan!70!black}{$\operatorname{rotate\_up}(A)$} (right).
Note here that in binary search trees, one typically calls the rotate operation on the node \((A\cup B)\), whereas in our case, that would be ambiguous. Instead, we define $\operatorname{rotate\_up}$ to take the child of \((A\cup B)\) whose depth is reduced.}
\end{figure}
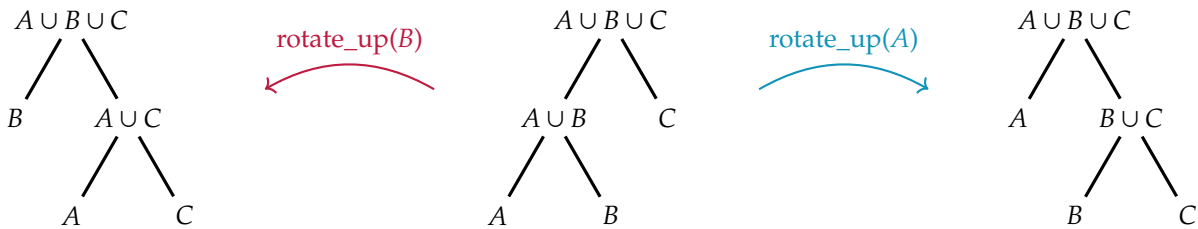

\item[semi\_splay\_step\((\mathtt{node})\)] This finds and executes one or
  two \texttt{rotate\_up} operations in the top tree that are valid and
  together reduce the depth of \texttt{node} by one. It returns the
  root of the changed subtree.

  This is really the core of our new algorithm. We prove that if the
  depth of \texttt{node} is at least five, this operation always
  succeeds and only needs to look at a small constant number of nodes.

  Conceptually, the semi-splay step serves the purpose that we use
  rotations for in binary search trees: decreasing the depth of the
  given node. We cannot use \texttt{rotate\_up} for that directly because we
  aren't always allowed to call that function. However,
  \texttt{rotate\_up} is allowed sufficiently often that we can still
  implement splay trees, and the semi-splay step is best thought of as a
  procedure that searches for such a rotation --- understanding the
  precise details of how that happens is not important for
  understanding the rest of the paper.

  This immediately suggests a na\"{\i}ve algorithm similar to the
  rotate-to-root algorithm by Allen et
  al.\cite{AllenMunro78rotatetoroot} for binary search trees that was
  a precursor to splay trees. Simply call \texttt{semi\_splay\_step}
  repeatedly on \texttt{node} until you get constant depth.  This,
  however does not give the desired amortized \(O(\log n)\) running
  time for our external operations.

\item[semi\_splay\((\mathtt{node})\)] This uses \texttt{semi\_splay\_step}
  repeatedly but in a slightly less na\"{\i}ve way.  The amortized
  cost of this operation is \(O(\log
  n)-\Omega(\depth(\mathtt{node}))\), which means that it can be used to
  pay for other operations whose natural cost is
  \(O(\depth(\mathtt{node}))\).

  This also guarantees that the depth of \texttt{node} is reduced by a
  constant factor, and is similar to the balancing operation used in
  semi-splay trees.

\item[full\_splay\((\mathtt{node})\)] This also uses
  \texttt{semi\_splay\_step} repeatedly. In addition to the guarantees
  that \texttt{semi\_splay} provides, this method also guarantees that
  \texttt{node} is moved to have depth at most $4$.

  It does this by considering two \texttt{semi\_splay\_step} calls at a
  time, and is similar to the logic in splay trees where we look for
  either a Zig, a Zig-Zig, or a Zig-Zag step.

\item[find\_consuming\_node\((v)\)] returns the \emph{consuming node}
  of a vertex, defined as the lowest common ancestor in the top tree
  of all edges incident to the vertex. More details and
  pseudocode can be found in Section~\ref{sec:consuming}. Note that
  although the query itself does not require any changes to the tree,
  our implementation will perform a semi-splay in the tree first to
  make the amortization work. This is similar to a search in a
  standard splay tree, where you also need to do a splay during or after
  each search.

\end{description}

\subsection{Data structure}
We store the top tree as a rooted binary tree. For each node, we store the following information:
\begin{itemize}
  \item A pointer to the parent node, or a null pointer if it has no parent.
  \item For internal nodes, a left and right child pointer. Neither pointer may
    be null.
  \item For leaf nodes, a vertex id or pointer for the left and right endpoints
    of the edge. (Or a pointer to the edge, which in turn stores the left and
    right endpoints.)
  \item A counter storing the number of boundary vertices of the node. This can
    be zero, one or two.
  \item Each node stores a flip bit for the subtree rooted at that
    node. It represents whether the subtree rooted at the node
    conceptually needs to be flipped relative to its parent.
\end{itemize}
It is also necessary to store the underlying tree itself. This can be done by
storing the edges adjacent to each vertex in a linked list. The necessary
operations on the underlying tree are: insert/delete edge, a way to get any edge
incident to a given vertex, a way to get the endpoints of an edge, and a way to
determine whether the degree of a vertex is ``zero or one'' or ``at least
two''. (We have queries like \texttt{degree(v) >= 2}, but we do not actually
need the precise number. This is convenient when storing the edges incident to a
vertex in a linked list, as we do not need a separate counter for its length.)

\subsection{User data}

Most applications of top trees require that you store some kind of user-data in
the vertices or edges in the underlying tree, and summaries based on that user-data in each cluster of the top tree. Normally, this
user-data must be computed ``bottom-up'', and the user-data usually doesn't
support changes in the middle of the top tree (e.g.\ rotations). However, it
turns out that this is not actually a problem for our operations, because they
always operate on the entire root path, and not only locally in the middle of the top
tree. This means that when modifying the top tree, you can first destroy the
user-data on the nodes of the root path, then you can run the algorithm from this paper, and
then you can rebuild the user-data on nodes of the new root path. Rebuilding can
be done either at the end of each operation, or deferred until the data is
actually needed. Either way, the amortized number of times node user-data must
be (re)computed stays \(O(\log n)\).

\section{Detecting boundary vertex positions}
\label{sec:has-x-boundary}

Our other internal operations need to determine whether the boundary vertices
are to the left, middle, or to the right. This can be done by looking at how
many boundary vertices the node and its children have. The operations can be
implemented as follows:

\pagebreak[3]

\begin{verbatim}
fn has_middle_boundary(node) {
    if is_leaf_node(node) || node.num_boundary_vertices == 0 {
        return false;
    } else {
        return node.num_boundary_vertices
            - is_path(node.children[0])
            - is_path(node.children[1]);
    }
}

fn has_left_boundary(node) {
    if is_leaf_node(node) {
        endpoint = node.endpoints[node.flip];
        return endpoint.exposed || degree(endpoint) >= 2;
    } else {
        return is_path(node.children[node.flip]);
    }
}

fn has_right_boundary(node) {
    if is_leaf_node(node) {
        endpoint = node.endpoints[!node.flip];
        return endpoint.exposed || degree(endpoint) >= 2;
    } else {
        return is_path(node.children[!node.flip]);
    }
}
\end{verbatim}
Checking whether a node has a left boundary vertex can be done by checking
whether the appropriate child is a path cluster because if the child is a point
cluster, then its only boundary vertex is the central vertex of our node, and if
the child is a path cluster, then it has a boundary vertex different from the
central vertex, and that extra boundary vertex is also a boundary vertex of the
parent. The code for \texttt{has\_middle\_boundary} works because
\(\mathtt{left} + \mathtt{middle} + \mathtt{right}\) is the total number of
boundary vertices. (Using the convention that booleans can be used
interchangeably with the integers zero and one.)

\section{Rotations}\label{sec:rotate}

The rotation is one of the basic operations that the top tree algorithms are
built on. To move a given node up, we swap it with its parent's sibling. The
rotate\_up operation is allowed if and only if \(\mathtt{sibling(node)} \cup
\mathtt{sibling(parent(node))}\) is a valid cluster. There are no guarantees
about which orientations the nodes are given, except that they satisfy the
orientation invariant.

Pseudocode for rotate\_up follows below:

\begin{verbatim}
fn rotate_up(node) {
    parent = node.parent
    grandparent = parent.parent
    sibling = sibling(node)
    uncle = sibling(parent)

    push_flip(grandparent)
    push_flip(parent)

    uncle_is_left_child = grandparent.children[0] == uncle
    sibling_is_left_child = parent.children[0] == sibling
    to_same_side = uncle_is_left_child == sibling_is_left_child
    sibling_is_path = is_path(sibling)
    uncle_is_path = is_path(uncle)
    gp_is_path = is_path(grandparent)

    if to_same_side && sibling_is_path {
        // Rotation on path.
        gp_middle = has_middle_boundary(grandparent)
        new_parent_is_path = gp_middle || uncle_is_path
        flip_new_parent = false
        flip_grandparent = false
        if gp_middle && !gp_is_path {
            ggp = grandparent.parent
            if ggp is not null {
                gp_is_left_child = ggp.children[0] == grandparent
                flip_grandparent = gp_is_left_child == uncle_is_left_child
            }
        }
    } else {
        // Rotation on star.
        if !to_same_side {
            new_parent_is_path = sibling_is_path || uncle_is_path
            flip_new_parent = sibling_is_path
            flip_grandparent = sibling_is_path
            node.flip = !node.flip
        } else {
            new_parent_is_path = uncle_is_path
            flip_new_parent = false
            flip_grandparent = false
            sibling.flip = !sibling.flip
        }
    }

    parent.children[uncle_is_left_child] = sibling
    parent.children[!uncle_is_left_child] = uncle
    parent.flip = flip_new_parent
    parent.boundary = if new_parent_is_path { 2 } else { 1 }

    grandparent.children[uncle_is_left_child] = node
    grandparent.children[!uncle_is_left_child] = parent
    grandparent.flip = flip_grandparent

    node.parent = grandparent
    uncle.parent = parent
}
\end{verbatim}

When performing a rotation, the underlying tree can look in one of two ways as
illustrated in Figure~\ref{fig:rotatetree}. Additionally, the orientations of
the top tree can appear with \texttt{node} and \texttt{sibling} to the same or
opposite sides as illustrated in Figure~\ref{fig:rotatetoptree}. Using this, we
can analyze the implementation of rotate\_up on a case-by-case basis.

First, we should argue that \texttt{to\_same\_sides
\&\& sibling\_is\_path} is true if and only if we are rotating around a path. If we
assume that we are rotating around a path, then \texttt{sibling} must be a path
cluster, and it follows from the orientation invariant that we are in the case
of Figure~\ref{fig:rotatesameside}. If we assume that we are rotating around a
star and \texttt{sibling} is a path cluster, then we can't be in case
Figure~\ref{fig:rotatesameside} because then \texttt{node} and \texttt{uncle}
would need to connect to opposite sides of \texttt{sibling} due to the orientation
invariant.

We now show correctness for rotations around a path like in
Figure~\ref{fig:rotatepath}. We know that we are in
Figure~\ref{fig:rotatesameside}, and the orientations are updated such that
\texttt{node}, \texttt{sibling}, and \texttt{uncle} appear in the same order
before and after the rotation. This means that the orientations match without
having to flip \texttt{node}, \texttt{sibling}, \texttt{uncle}, or
\texttt{parent}. The line that computes \texttt{new\_parent\_is\_path} needs to be true when
\(\{\mathtt{sibling}, \mathtt{uncle}\}\) has two boundary vertices. That cluster
always has the boundary vertex between \texttt{sibling} and \texttt{node}, so we
just need to check if one of the two other vertices are boundary vertices. The
vertex between \texttt{sibling} and \texttt{uncle} is a boundary vertex of the new
\texttt{parent} if and only if it is a boundary vertex of the old
\texttt{grandparent}, and it is the central vertex of the old
\texttt{grandparent}, so we can check that by asking whether
\texttt{grandparent} had a middle boundary vertex. The vertex at the other end
of \texttt{uncle} is a boundary of \(\{\mathtt{sibling}, \mathtt{uncle}\}\)
whenever it is a boundary vertex of \texttt{uncle}. We can check this by asking
whether \texttt{uncle} has two boundary vertices.

The code that looks at \texttt{ggp} checks for the situation where
\(\{\mathtt{node}, \mathtt{sibling}, \mathtt{uncle}\}\) has a single boundary
vertex in the middle before the rotation, and a single boundary vertex to the
left or right after the rotation. If the new boundary vertex is to the right,
then the grandparent no longer has a leftmost boundary vertex, so if the
grandparent is a right child, we need to flip it to avoid breaking the
orientation invariant at \texttt{ggp}.

We now consider rotations on a star like in Figure~\ref{fig:rotatestar}. We
first consider the case in Figure~\ref{fig:rotatesameside}, where \texttt{sibling}
must be a point cluster to avoid violating the orientation invariant. Here, we
put the nodes together such that \texttt{node}, \texttt{sibling}, and
\texttt{uncle} appear in the same order before and after the rotation. We flip
\texttt{sibling} because it changes between being a left and right child. (The flip
isn't necessary when \texttt{sibling} has the boundary vertex in the middle, but it
is also not incorrect to flip in this case.) The new cluster is a path cluster
exactly when \texttt{uncle} is a path cluster since we know that \texttt{sibling}
is a point cluster. The orientation of the grandparent remains correct since the
rotation does not change which vertex is to the left, in the middle, and to the
right.

\begin{figure}
  \begin{subfigure}{0.6\textwidth}
    \begin{tikzpicture}
      \begin{scope}
        \node[circle,draw] (v1) at (0,0) {};
        \node[circle,draw] (v2) at (2,1) {};
        \node[circle,draw] (v3) at (4,0) {};
        \node[circle,draw] (v4) at (6,1) {};
      \end{scope}
      \begin{scope}[every edge/.style={draw}]
        \path (v1) edge node[sloped,above] {\texttt{node}} (v2);
        \path (v2) edge node[sloped,above] {\texttt{sibling}} (v3);
        \path (v3) edge node[sloped,above] {\texttt{uncle}} (v4);
      \end{scope}
    \end{tikzpicture}
    \caption{Rotation on a path.}
    \label{fig:rotatepath}
  \end{subfigure}%
  \begin{subfigure}{0.4\textwidth}
    \begin{tikzpicture}
      \begin{scope}
        \node[circle,draw] (v1) at (0,0) {};
        \node[circle,draw] (v2) at (2,1) {};
        \node[circle,draw] (v3) at (4,0) {};
        \node[circle,draw] (v4) at (2,{1+sqrt(5)}) {};
      \end{scope}
      \begin{scope}[every edge/.style={draw}]
        \path (v1) edge node[sloped,above] {\texttt{sibling}} (v2);
        \path (v2) edge node[sloped,above] {\texttt{node}} (v3);
        \path (v2) edge node[sloped,above] {\texttt{uncle}} (v4);
      \end{scope}
    \end{tikzpicture}
    \caption{Rotation on a star.}
    \label{fig:rotatestar}
  \end{subfigure}
  \caption{The two possible configurations in the underlying tree of the
  clusters involved in a rotation. Each cluster is drawn as an edge. The labels
  on the edges refers to variables in the pseudocode for \texttt{rotate\_up}.}
  \label{fig:rotatetree}
\end{figure}
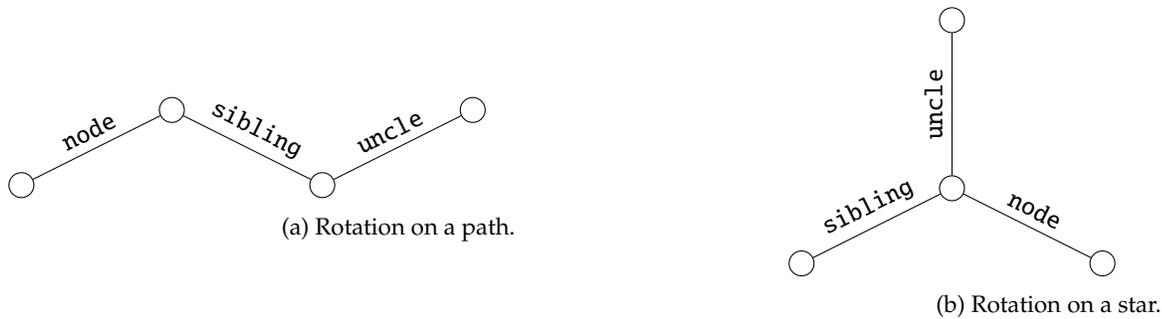

\begin{figure}
  \begin{subfigure}{0.5\textwidth}
    \begin{tikzpicture}
      \begin{scope}
        \node (c0) at (2,0) {\texttt{sibling}};
        \node (c1) at (0,0) {\texttt{node}};
        \node (c2) at (3,{1*sqrt(3)}) {\texttt{uncle}};
        \node (b1) at (1,{1*sqrt(3)}) {\texttt{parent}};
        \node (b2) at (2,{2*sqrt(3)}) {\texttt{grandparent}};
      \end{scope}

      \begin{scope}[every edge/.style={draw=black,very thick},]
        \path (c0) edge (b1);
        \path (c1) edge (b1);
        \path (c2) edge (b2);
        \path (b1) edge (b2);
      \end{scope}
    \end{tikzpicture}
    \caption{\texttt{sibling} and \texttt{uncle} to the same side.}
    \label{fig:rotatesameside}
  \end{subfigure}%
  \begin{subfigure}{0.5\textwidth}
    \begin{tikzpicture}
      \begin{scope}
        \node (c0) at (2,0) {\texttt{node}};
        \node (c1) at (0,0) {\texttt{sibling}};
        \node (c2) at (3,{1*sqrt(3)}) {\texttt{uncle}};
        \node (b1) at (1,{1*sqrt(3)}) {\texttt{parent}};
        \node (b2) at (2,{2*sqrt(3)}) {\texttt{grandparent}};
      \end{scope}

      \begin{scope}[every edge/.style={draw=black,very thick},]
        \path (c0) edge (b1);
        \path (c1) edge (b1);
        \path (c2) edge (b2);
        \path (b1) edge (b2);
      \end{scope}
    \end{tikzpicture}
    \caption{\texttt{sibling} and \texttt{uncle} to different sides.}
    \label{fig:rotateopposide}
  \end{subfigure}
  \caption{The two possible configurations in the top tree of \texttt{sibling} and
  \texttt{uncle} in relation to each other (ignoring reflections). The labels on
  the edges refers to variables in the pseudocode for \texttt{rotate\_up}.}
  \label{fig:rotatetoptree}
\end{figure}
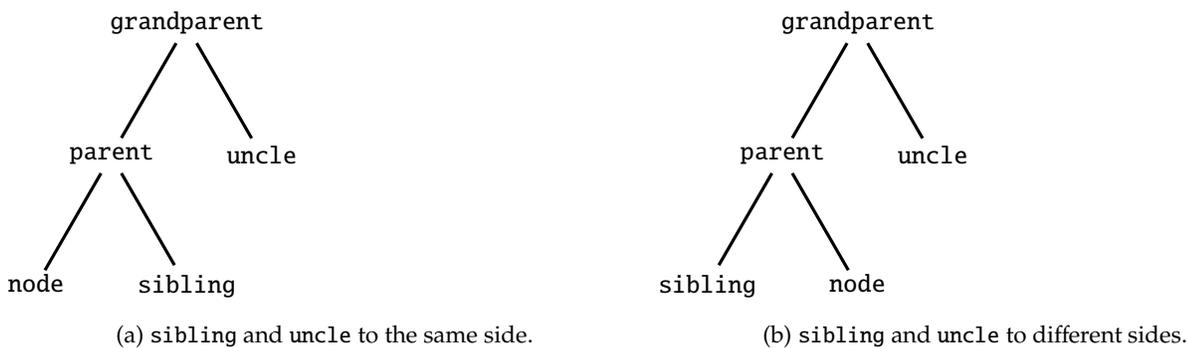

Next, we consider rotations on a star like in
Figure~\ref{fig:rotateopposide} with \texttt{sibling} being a point cluster. Here,
\texttt{node} must be a point cluster to avoid violating the orientation
invariant. The only two possible boundary vertices of the new \texttt{parent}
and \texttt{grandparent} are the two endpoints of \texttt{uncle}. Thus, the new
\texttt{parent} is a path cluster when \texttt{uncle} is. To avoid breaking the
orientation invariant above \texttt{grandparent}, we note that the code puts the
tree back together such that the two endpoints of \texttt{uncle} do not change
positions in \texttt{grandparent}.

The last case to consider is Figure~\ref{fig:rotateopposide} with \texttt{sibling}
being a path cluster. Then, both \texttt{node} and \texttt{uncle} must be
point clusters, because otherwise we have a cluster with three boundary vertices
before or after the rotation. Thus, the only two possible boundary vertices of
the new \texttt{parent} and \texttt{grandparent} are the two endpoints of
\texttt{sibling}. Thus, \texttt{parent} is a path cluster when \texttt{sibling} is. To
avoid breaking the orientation invariant above \texttt{grandparent}, we note
that the code puts the tree back together such that the two endpoints of
\texttt{sibling} do not change positions in \texttt{grandparent}.

The above analysis was based on the two cases in Figure~\ref{fig:rotatetoptree},
but the two cases can also appear flipped. The algorithm also works in those
cases because the code only cares about whether things hang off to the same or
different sides, and not which side is left or right, so it treats the flipped
cases identically.

\subsection{When are rotations valid?}
\label{sec:rotvalid}

The exact condition for when rotates are allowed is that
\(\mathtt{sibling(node)} \cup \mathtt{sibling(parent(node))}\) is a valid
cluster, but proving this directly every time we wish to make a rotation quickly
becomes quite cumbersome. For this reason, we provide two lemmas that provide
some cases in which rotations are allowed.

We note that since \texttt{rotate\_up} may modify the orientations in any way it
wants to as long as the invariants are satisfied, there are multiple correct
ways to implement it. However, the lemmas below do not look inside our
implementation, so they hold for any correct implementation of
\texttt{rotate\_up}.

However, we first prove some helper lemmas.
\begin{lemma}
  \label{lem:invalid-iff}
  If \(a\) and \(b\) are valid clusters of a top tree whose intersection is a
  single vertex \(v\), then the cluster \(a \cup b\) is invalid if and only if
  the following three conditions hold: \(a\) is a path cluster, \(b\) is a path
  cluster, \(v\) is a boundary vertex of \(a \cup b\).
\end{lemma}
\begin{proof}
  We have already seen that the boundary vertices of a cluster is the union of
  the boundary vertices of its children, possibly with the central vertex
  removed. Since \(a\) and \(b\) are valid and share \(v\), the union of the
  boundary vertices of the children has three elements if and only if both \(a\)
  and \(b\) are path clusters. The central vertex of \(a \cup b\) is \(v\). The
  lemma follows from this.
\end{proof}

\begin{lemma}
  \label{lem:point-boundary}
  If \(a,b,c\) are valid clusters with \(c\) the parent of \(a\) and \(b\), and
  if \(a\) is a point cluster, then all boundary vertices of \(c\) are also
  boundary vertices of \(b\).
\end{lemma}
\begin{proof}
  All boundary vertices of \(c\) are either boundary vertices of \(a\) or \(b\).
  However, all boundary vertices of \(a\) are also boundary vertices of \(b\)
  because the only boundary vertex of a point cluster is the one it shares with
  its sibling.
\end{proof}

In the following we define the phrase ``\(a\) and \(b\) hang off to the same
side'' as ``both \(a\) and \(b\) are left children, or both are right
children''. The phrase ``\(a\) and \(b\) hang off to opposite sides'' will be
its negation.
\begin{lemma}
  \label{lem:path-not-con}
  Let \(x,y,z,a,b\) be valid clusters in a top tree with \(y\) the parent of
  \(x,a\) and \(z\) the parent of \(y,b\). If \(a\) is a path cluster and \(x,y\)
  hang off to the same side, then \(x \cup b\) is not a connected set of edges.
\end{lemma}
\begin{proof}
  Assume without loss of generality that \(a,b\) are right children. Since \(a\)
  is a path cluster, it has a leftmost boundary vertex \(v\) and a rightmost
  boundary vertex \(w\). By the orientation invariant, \(y\) has central vertex
  \(v\) and right boundary vertex \(w\). Similarly, \(w\) is the central
  vertex of \(z\). It follows that \(x\) has boundary vertex \(v\)
  and that \(b\) has boundary vertex \(w\). If \(x \cup b\) was a connected set
  of edges, then they would share a vertex \(u\), but this is impossible as
  there would be a path from \(u\) to \(w\) through \(b\), then from \(w\) to
  \(v\) through \(a\), then from \(v\) to \(u\) through \(x\). This is a cycle
  in a tree. Thus, \(x \cup b\) is not connected.
\end{proof}

Using these helper lemmas, we can prove our lemmas for when rotations are
allowed. The variable names in the lemmas are defined to match
Figure~\ref{fig:rotvalid}.

\begin{lemma}
  \label{lem:point-rotate}
    Let \(x\) be a valid cluster in a top tree. If \(x\) and its grandparent are
    both point clusters, then it is valid to call
    \(\operatorname{rotate\_up}(x)\).
\end{lemma}
\begin{proof}
  Let \(y\) be the parent of \(x\) and let \(z\) be the parent of \(y\). Let
  also \(a\) be the sibling of \(x\) and \(b\) be the sibling of \(y\). We need
  to prove that \(a \cup b\) is a valid cluster. To do that, we first argue that
  \(a\) and \(b\) share a vertex \(v\). By Lemma~\ref{lem:point-boundary}, all
  boundary vertices of \(y\) are also boundary vertices of \(a\). However, one
  of the boundary vertices of \(y\) is the one it shares with its sibling \(b\).
  Thus, \(a\) and \(b\) share a boundary vertex.

  It remains to show that \(a \cup b\) is valid. Since \(x \cup a\) already
  exists in the tree, the only boundary vertex in \(x\) must be the one it
  shares with \(a\). Thus, the cluster \((a \cup b) \cup x\) has at most one
  fewer boundary vertices than \(a \cup b\). Since \((a \cup b) \cup x\) already
  exists in the tree as the point cluster \(z\), this means that \(a \cup b\)
  has at most two boundary vertices as desired.
\end{proof}

\begin{lemma}
  \label{lem:path-rotate}
  Let \(x,y,z,a,b\) be valid clusters in a top tree with \(z\) the parent of
  \(y,b\) and \(y\) the parent of \(x,a\). If \(y\) is a path cluster and \(x\)
  hangs off to the same side as \(y\), then it is valid to call
  \(\operatorname{rotate\_up}(x)\).

  If \(z\) is a path cluster and not the root, then \(a \cup b\) hangs off to
  the same side as \(z\) after the rotation if and only if \(b\) hung off to the
  same side as \(z\) before the rotation.

  If \(z\) is a point cluster, then \(a \cup b\) is a point cluster after the
  rotation.
\end{lemma}
\begin{proof}
  We need to show that \(a \cup b\) is a valid cluster. To show that it is
  connected, let \(v\) be the central vertex of \(z\). It is a boundary vertex
  of its children \(y\) and \(b\). Since \(y\) is a path cluster, it has another
  boundary vertex \(w\), and due to the orientation invariant, the assumption
  that \(x,y\) hang off to the same side implies that \(w\) is a boundary vertex
  of \(x\) and \(v\) is a boundary vertex of \(a\). Thus, since \(a\) and \(b\)
  share the vertex \(v\), they are connected.

  Let us show that \(a \cup b\) is a valid cluster. If \(a\) or \(b\) are point
  clusters, then we are done by Lemma~\ref{lem:invalid-iff}, so assume both are
  path clusters. Since \(y\) is also a path cluster, it follows from
  Lemma~\ref{lem:invalid-iff} that \(v\) is not a boundary vertex of \(z\).
  Since \(z = x \cup (a \cup b)\), it follows that \(v\) can only be a boundary
  vertex of \(a \cup b\) if \(v\) is the vertex shared with \(x\), but \(x \cup
  b\) is disconnected by Lemma~\ref{lem:path-not-con}, making it impossible for
  \(v\) to be a boundary vertex of \(x\). Thus, \(v\) is not a boundary vertex
  of \(a \cup b\), so it is valid by Lemma~\ref{lem:invalid-iff}.

  Assume that \(z\) is a path cluster and that its sibling \(c\) exists. The
  cluster \(z\) has a boundary vertex that isn't in the middle. Either it comes
  from \(x \cup a\) before and \(x\) after, or it comes from \(b\) before and
  \(a \cup b\) after. If that boundary vertex is shared with \(c\), then the
  child it comes from must hang off to the same side as \(c\) both before and
  after the rotation. If it isn't shared with \(c\), then the child it comes
  from must hang off to the opposite side as \(c\) both before and after the
  rotation. The desired post-condition follows.

  Assume that \(z\) is a point cluster. Since \(x,y\) hang off to the same side,
  the only boundary vertex of \(z\) must be \(w\). However, since \(w\) is not
  in the middle at \(z\), it must come exclusively from one child, and we know
  that it is a boundary vertex of \(x\), so it must come from \(x\). This
  implies that \(a \cup b\) cannot be a path cluster after the rotation.
\end{proof}

\begin{figure}
  \centering
  \begin{tikzpicture}[scale=0.5]
    \begin{scope}
      \node (c0) at (2,0) {\(a\)};
      \node (c1) at (0,0) {\(x\)};
      \node (c2) at (3,{1*sqrt(3)}) {\(b\)};
      \node (b1) at (1,{1*sqrt(3)}) {\(y\)};
      \node (b2) at (2,{2*sqrt(3)}) {\(z\)};
    \end{scope}

    \begin{scope}[every edge/.style={draw=black,very thick},]
      \path (c0) edge (b1);
      \path (c1) edge (b1);
      \path (c2) edge (b2);
      \path (b1) edge (b2);
    \end{scope}
  \end{tikzpicture}
  \caption{An illustration of a tree with the variable names used in
  Section~\ref{sec:rotvalid}.}
  \label{fig:rotvalid}
\end{figure}
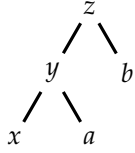

\section{Splaying}
\label{sec:splay}
Splaying is a well-known strategy for balancing binary search trees, and typically
comes in two variants: full splays and semi-splays. Splaying does not translate
verbatim to top trees since not all rotations are allowed in top trees, but it
turns out that rotations are allowed sufficiently often that we can still
implement splay-like operations. This section will describe both a semi-splay
and a full splay, where the semi-splay will reduce the depth of the given node
by a constant factor, and the full splay will reduce the depth to four or less.

If we only consider the asymptotic amortized running times in terms of \(n\),
then there is no reason to ever use a semi-splay, since the full splay provides
strictly stronger guarantees and has the same amortized running time
asymptotically. However, since semi-splays make fewer structural updates to the
tree, they are generally faster by a constant factor, which we conjecture may be significant
in practice. In this paper, we will only use a full splay when we need the
additional guarantees it provides.

Our proofs that the operations run in amortized logarithmic time use the
amortization potential of \(\Phi = \sum_T\sum_{x \in T} r(x)\), with \(r(x) =
\log_2(s(x))\) and \(s(x)\) the number of leaves in the subtree rooted at \(x\).
This is similar to the potential used for splay trees, which uses the number of
nodes rather than leaves in the subtree.

\subsection{Semi-splay step}

\todo[inline]{R1: A figure visualizing what can happen in \texttt{splay\_step()}
would be very useful, maybe in the appendix. It was impossible for me to
understand what's happening without sketching every situation myself.}

To implement this, we will first define a semi-splay step. Implementing a semi-splay
or a full splay involves repeatedly calling the semi-splay step operation.
Conceptually, the semi-splay step serves the purpose that we use rotations for in
binary search trees: decreasing the depth of the given node. We cannot use
\texttt{rotate\_up} for that directly because we aren't always allowed to call that
function. However, \texttt{rotate\_up} is allowed sufficiently often that we can
still implement splay trees, and the semi-splay step is best thought of as a
procedure that searches for such a rotation --- understanding the precise
details of how that happens is not important for understanding the rest of the
paper.

The pseudo-code for a semi-splay step is found below:

\begin{verbatim}
fn semi_splay_step(node) {
    parent = node.parent
    gparent = parent.parent
    if parent or gparent is null {
        return null
    }

    if is_point(node) && is_point(gparent) {
        rotate_up(node)
        return gparent
    }

    ggparent = gparent.parent
    if ggparent is null {
        return null
    }
    if is_path(parent) && (is_path(gparent) || is_point(ggparent)) {
        push_flip(gparent)
        push_flip(parent)

        node_is_left = parent.children[0] == node
        parent_is_left = gparent.children[0] == parent
        gparent_is_left = ggparent.children[0] == gparent
        if node_is_left == parent_is_left {
            rotate_up(node)
            return gparent
        }
        if parent_is_left == gparent_is_left {
            rotate_up(parent)
            return ggparent
        }
        // At this point `node_is_left == gparent_is_left' is true.
        rotate_up(sibling(node)) // swaps sibling(node) and sibling(parent)
        rotate_up(parent) // parent is still node.parent here
        return ggparent
    }
    return semi_splay_step(parent)
}
\end{verbatim}

Since rotations are not always valid, the \texttt{semi\_splay\_step} operation will
follow the root path \(b_0, b_1, \dots, b_k\) until it matches a pattern where
the rotation is allowed. Once the pattern is detected, the algorithm knows that
it has found two nodes hanging off the root path that can be merged, and it
merges them using rotations. If the pattern doesn't match, it tries again by
calling itself with \(b_0\) being the previous \(b_1\). The pattern is such that
it can only fail a few times, and the first node in the match will be one of
\(b_0\), \(b_1\), \(b_2\), or \(b_3\).

The patterns that the algorithm will match are:
\begin{itemize}
  \item \(b_0b_1b_2 = \mathtt{point}, \mathtt{point}, \mathtt{point}\)
  \item \(b_0b_1b_2 = \mathtt{point}, \mathtt{path}, \mathtt{point}\)
  \item \(b_1b_2b_3 = \mathtt{path}, \mathtt{path}, \mathtt{point}\)
  \item \(b_1b_2b_3 = \mathtt{path}, \mathtt{point}, \mathtt{point}\)
  \item \(b_1b_2b_3 = \mathtt{path}, \mathtt{path}, \mathtt{path}\)
\end{itemize}
Here, some of the patterns are matched on \(b_0b_1b_2\), and some are matched on
\(b_1b_2b_3\). The semi-splay step prefers to pick a pattern on \(b_0b_1b_2\) if
there is a match for both types.

\begin{lemma}
  A semi-splay step never makes an invalid rotation in a valid top tree.
\end{lemma}
\begin{proof}
  If the semi-splay step hits the \texttt{is\_point(node) \&\& is\_point(gparent)}
  branch, then the rotation is allowed by Lemma~\ref{lem:point-rotate}.
  If we hit the \texttt{node\_is\_left == parent\_is\_left} branch, then since
  \texttt{parent} is a path cluster, the rotation is valid by
  Lemma~\ref{lem:path-rotate}.

  If we hit the \texttt{parent\_is\_left == gparent\_is\_left} branch, then it
  must be the case that \texttt{gparent} is a path cluster. To see this, assume for
  contradiction that it is a point cluster, and assume without loss of
  generality that \texttt{parent} is a left child. Then \texttt{gparent} must
  have a left boundary vertex. This means that \texttt{gparent} has no rightmost
  boundary vertex, but it is a left child so its rightmost boundary vertex must
  exist. Thus, \texttt{gparent} is a path cluster, and the rotation is valid by
  Lemma~\ref{lem:path-rotate}.

  If we hit the \texttt{node\_is\_left == gparent\_is\_left} branch, then since
  \texttt{parent} is a path cluster, the first rotation is valid by
  Lemma~\ref{lem:path-rotate}. If \texttt{gparent} is a path cluster (the
  rotation doesn't change this), then the first post-condition of
  Lemma~\ref{lem:path-rotate} says that \texttt{parent} and \texttt{gparent}
  hang off to the same side after the rotation, so Lemma~\ref{lem:path-rotate}
  says that the second rotation is valid. If \texttt{gparent} is a point
  cluster, then the second post-condition says that \texttt{parent} is a point
  cluster after the rotation, and this case is not reachable unless
  \texttt{ggparent} is also a point cluster, so the second rotation is allowed
  by Lemma~\ref{lem:point-rotate}.
\end{proof}

\begin{lemma}\label{lem:semi-splay-step-success}
  If $x$ has depth $d\geq 5$, then $\mathtt{semi\_splay\_step}(x)$ reduces
  the depth of $x$ by $1$, and returns a node $x'$ with $d-5\leq
  \depth(x')\leq d-2$ such that all modified nodes are in the subtree
  rooted at $x'$.
\end{lemma}
\begin{proof}
  To see that the semi-splay step will match on \(b_3b_4b_5\) or earlier, note that if
  it fails at \(b_1b_2b_3\) and \(b_2b_3b_4\), then \(b_3\) and \(b_4\) are both
  path clusters, but then it must succeed at \(b_3b_4b_5\). The \(b_0b_1b_2\)
  patterns always returns \(b_2\) and the \(b_1b_2b_3\) patterns always returns
  \(b_2\) or \(b_3\). This implies that the returned node is one of
  \(b_2,b_3,b_4,b_5\) which have depth as desired. The remainder of the lemma
  follows immediately by considering each case separately.
\end{proof}

\pagebreak[2]

\begin{lemma}\label{lem:semi-splay-step-fail}
  If $\mathtt{semi\_splay\_step}(x)$ returns \texttt{null}, then no change
  was made to the tree and \(\depth(x) \leq 4\). If \(x\) is a point cluster,
  then \(\depth(x) \leq 3\).  If the root is a point cluster, then \(\depth(x)
  \leq 2\). If both \(x\) and the root are point clusters, then \(\depth(x) \leq
  1\).
\end{lemma}

\begin{proof}
  We get \(\depth(x) \leq 4\) since the semi-splay step always succeeds for
  \(\depth(x) \geq 5\) by Lemma~\ref{lem:semi-splay-step-success}. If the root
  cluster is a point cluster, then since we match all patterns ending in a point
  cluster, the failure to match must be because \(b_3\) doesn't exist, so
  \(\depth(x) \leq 2\). If \(x\) is also a point cluster, then the failure to
  match must be because \(b_2\) doesn't exist, so \(\depth(x) \leq 1\).
  Finally, if \(x\) is a point cluster and the root is a path cluster, then in
  every possible pattern for \(b_0b_1b_2b_3b_4\) either both of \(b_0,b_2\) are
  point clusters, or \(b_3\) is a point cluster, or both of \(b_2,b_3\)
  are path clusters. In each case \(\mathtt{semi\_splay\_step}(x)\) would return
  non-null, so \(b_4\) can't exist and \(\depth(x) \leq 3\).
\end{proof}

\begin{lemma}\label{lem:semi-splay-step-pot}
  Consider the operation \(\mathtt{semi\_splay\_step}(p_k(x))\) with
  \(p_i(x)\) the \(i\)th ancestor of \(x\) or the root if \(i \geq
  \depth(x)\). Let \(x'\) be the node \(x\) after the semi-splay step. Let
  \(0 \leq k \leq b\) and \(0\leq a \leq k+1\) be such that
  \(\mathtt{semi\_splay\_step}(p_k(x))\) returns \(p_\ell(x')\) for some \(\ell
  \leq b\) (this is the same node as \(p_{\ell+1}(x)\)). Then, the splay
  step changes the amortization potential by \(\Delta\Phi =
  r(\sibling(p_{\ell-1}(x'))) + \sum_{i=a}^{b-1} r(p_i(x')) -
  \sum_{i=a}^b r(p_i(x))\).
\end{lemma}
\begin{proof}
  The semi-splay step either rotates once or twice. By drawing the tree before and
  after, we can see that if it rotates once, then the potential changes by
  \(\Delta\Phi = r(\sibling(p_{\ell-1}(x'))) - r(p_\ell(x))\). If it rotates twice,
  then the potential changes by \(\Delta\Phi = r(\sibling(p_{\ell-1}(x'))) +
  r(p_{\ell-1}(x')) - r(p_\ell(x)) - r(p_{\ell-1}(x))\). These expressions are
  equal to the expression in the Lemma because the remaining terms are just the
  same nodes added and subtracted before and after the operation, but none of
  those nodes had their number of leaves changed, so the extra terms cancel out.
\end{proof}

\subsection{Semi-splay}
Using \texttt{semi\_splay\_step} as a subroutine, the pseudocode of
\texttt{semi\_splay} is:
\begin{verbatim}
fn semi_splay(node) {
    top = node
    while top is not null {
        top = semi_splay_step(top)
    }
}
\end{verbatim}
This subroutine will reduce the depth of \texttt{node} by a constant factor. It
does this by repeatedly calling \texttt{semi\_splay\_step} on an ancestor of
\texttt{node}. It is important for the amortized analysis that the semi-splay steps
do not ``overlap'', which the algorithm avoids because all nodes that the
\texttt{semi\_splay\_step} modifies are descendants of the node that it returns.

We will need the following small lemma in our analysis:
\begin{lemma}\label{lem:log-sum}
  For $a,b,c>0$, if $a+b\leq c$ then $\log_2(a)+\log_2(b)\leq 2\log_2(c) - 2$.
\end{lemma}
\begin{proof}
  Let $a,b,c>0$ and $a+b\leq c$, then by the arithmetic-geometric mean
  inequality we have $\sqrt{ab}\leq\frac{a+b}{2}\leq\frac{c}{2}$, thus
  $ab\leq(\frac{c}{2})^2$ and $\log_2(a)+\log_2(b)=\log_2(ab)\leq
  2\log_2(\frac{c}{2})=2(\log_2(c)-1)$.
\end{proof}
\begin{lemma}\label{lem:semi-splay-pot}
  Calling \texttt{semi\_splay} on a node $x$ does
  $O(\depth(x))$ work, changes the potential by
  $\Delta\Phi\leq O\paren[\big]{1+r(\treeroot(x))-r(x)} -
  \Omega(\depth(x))$, and reduces the depth of $x$ to at most
  $\ceil*{\tfrac{4}{5}\depth(x)}$.
\end{lemma}
In other words, the amortized cost of \texttt{semi\_splay} is
$O\paren[\big]{1+r(\treeroot(x))-r(x)}=O(\log n)$.

\begin{proof}
  We consider an iteration of the \texttt{semi\_splay}. Let \(x\) be
  \texttt{top} before the semi-splay step, and let \(x'\) be the same node in the
  tree after the semi-splay step. Then, the next value of \texttt{top} is
  \(p_\ell(x')\) for some integer \(1 \leq \ell \leq 4\). By
  Lemma~\ref{lem:semi-splay-step-pot} (with \(a=k=0\) and \(b=\ell\)), the potential
  changes by \(\Delta\Phi = r(\sibling(p_{\ell-1}(x'))) + \sum_{i=0}^{\ell-1}
  r(p_i(x')) - \sum_{i=0}^{\ell} r(p_i(x))\). By Lemma~\ref{lem:log-sum}, we get
  \(r(p_{\ell-1}(x')) + r(\sibling(p_{\ell-1}(x'))) \leq 2r(p_\ell(x')) - 2\),
  from which it follows that \(\Delta\Phi \leq 2r(p_\ell(x')) +
  \sum_{i=0}^{\ell-2} r(p_i(x')) - \sum_{i=0}^{\ell} r(p_i(x)) - 2\), and since
  \(\ell-2 \leq 2\), we get \(\Delta\Phi \leq 5r(p_\ell(x')) - 5r(x) - 2 =
  5r(\text{new \texttt{top}}) - 5r(\text{old \texttt{top}}) - 2\).

  Since the semi-splay step is always successful when \(\depth(\mathtt{top}) \geq
  5\), we must have at least \(\floor*{\tfrac{1}{5}\depth(x)}\geq
  \tfrac{1}{5}(\depth(x)-4)\) successful iterations. Summing up over all iterations,
  the total change in potential is therefore at most \(\Delta\Phi\leq
  5\paren[\big]{r(\treeroot(x))-r(x)} - \tfrac{2}{5}(\depth(x)-4) =
  5\paren[\big]{\tfrac{8}{25}+r(\treeroot(x))-r(x)} - \tfrac{2}{5}\depth(x)\). This is as
  desired.

  Each successful iteration reduces the depth of \(x\) by \(1\), so the
  resulting depth of \(x\) is at most \(\depth(x) -
  \floor*{\tfrac{1}{5}\depth(x)} = \ceil*{\tfrac{4}{5}\depth(x)}\) as claimed.
\end{proof}

\subsection{Full splay}
Using \texttt{semi\_splay\_step} as a subroutine, the pseudocode for
\texttt{full\_splay} is:
\begin{verbatim}
fn full_splay(node) {
    while true {
        top = semi_splay_step(node)
        if top is null {
            return
        }
        semi_splay_step(top)
    }
}
\end{verbatim}
This subroutine will move \texttt{node} close to the root so that its depth is
bounded by a constant. It does this by alternating between calling
\texttt{semi\_splay\_step} on \texttt{node} and on one of its ancestors \texttt{top}.
If we just wanted to reduce the depth to a constant, then calling
\texttt{semi\_splay\_step(node)} in a loop would suffice, but the second
\texttt{semi\_splay\_step} is necessary to make the amortized running time work. This
is similar to ordinary splay trees, which require Zig-Zig or Zig-Zag steps
rather than simply repeating Zig steps.

\begin{lemma}\label{lem:full-splay}
  Calling \texttt{full\_splay} on a node $x$ does $O(\depth(x))$ work,
  changes the potential by $\Delta\Phi\leq O\paren[\big]{1+r(\treeroot(x))-r(x)} -
  \Omega(\depth(x))$, and reduces the depth of $x$ so it satisfies the same
  bounds as in Lemma~\ref{lem:semi-splay-step-fail}.
\end{lemma}
In other words, the amortized cost of \texttt{full\_splay} is
$O\paren[\big]{1+r(\treeroot(x))-r(x)}=O(\log n)$.

\begin{proof}
  We consider the change of potential in an iteration of \texttt{full\_splay}
  where both semi-splay steps succeed. We let \(x\) be the node before the iteration,
  \(x'\) the node after one semi-splay step, and \(x''\) the node after two splay
  steps. Let \(p_{\ell_1}(x')\) and \(p_{\ell_2}(x'')\) be the return values of each
  semi-splay step. By Lemma~\ref{lem:semi-splay-step-success}, we have \(1 \leq \ell_1 \leq
  4\) and \(\ell_1 < \ell_2 \leq 8\). By using Lemma~\ref{lem:semi-splay-step-pot} twice
  (first with \((a,k,b)=(1,0,9)\) then with \((a,k,b)=(1,\ell_1,8)\)), the two
  semi-splay steps change the potential by
  \begin{equation}
    \Delta\Phi = r(\sibling(p_{\ell_2-1}(x''))) + r(\sibling(p_{\ell_1-1}(x')))
    + \sum_{i=1}^{7} r(p_i(x'')) - \sum_{i=1}^{9} r(p_i(x)).
  \end{equation}
  Note that the term \(\sum_{i=1}^{8} r(p_i(x'))\) has canceled in the above.
  Since \(\sibling(p_{\ell_2-1}(x''))\) and \(\sibling(p_{\ell_1-1}(x'))\) have
  disjoint subtrees and are both descendants of \(p_8(x'')\), it follows by
  Lemma~\ref{lem:log-sum} that \(r(\sibling(p_{\ell_2-1}(x''))) +
  r(\sibling(p_{\ell_1-1}(x'))) \leq 2r(p_8(x''))-2\). We note further that
  \(2r(p_8(x''))-2 \leq  r(p_8(x''))+r(p_{9}(x''))-2\), thus \(\Delta\Phi \leq
  \sum_{i=1}^{9} r(p_i(x'')) - \sum_{i=1}^{9} r(p_i(x)) - 2\). Analogously, the
  iterations where only one of the semi-splay steps succeed increase the potential by
  at most \(\Delta\Phi \leq \sum_{i=1}^{9} r(p_i(x'')) - \sum_{i=1}^{9}
  r(p_i(x))\), without the minus two term.

  Since both semi-splay steps always succeed when \(\depth(x) \geq 9\) and the depth
  decreases by two each time, both semi-splay steps succeed at least
  \(\tfrac{1}{2}(\depth(x)-9)\) times. Summing the changes of potentials over
  all iterations telescopes to \(\Delta\Phi \leq \sum_{i=1}^{9} r(p_i(x')) -
  \sum_{i=1}^{9} r(p_i(x)) - (\depth(x)-9) \leq 9\paren[\big]{1 +
  r(\treeroot(x))-r(x)} - \depth(x)\) with \(x\) the node before and \(x'\) the
  node after the entire full splay. This is as desired. That the depth bounds
  are as desired follows by Lemma~\ref{lem:semi-splay-step-fail} since the semi-splay step
  has just failed on \(x'\) when the algorithm returns.
\end{proof}

\section{Finding the consuming node of a vertex}\label{sec:consuming}

The consuming node of a vertex is defined as the lowest common ancestor in the
top tree of all edges incident to the vertex. When the vertex is not exposed,
this is the same as the smallest cluster containing the vertex, without having
the vertex as a boundary vertex. Since boundary vertices can only ``disappear''
in nodes where they are the central vertex, the consuming node for a non-leaf vertex is also the
largest cluster with the vertex as the central vertex. The latter condition also
works when the vertex is exposed.

Pseudocode for finding the consuming node is as follows:
\begin{verbatim}
fn find_consuming_node(v) {
    if v has no neighbors {
        return null
    }
    node = any edge incident to v
    semi_splay(node)
    if v has at most one incident edge {
        return node
    }

    // Where is v a boundary vertex in node? (taking flip into account)
    is_left = (is v left endpoint of node) != node.flip
    is_middle = false
    is_right = (is v right endpoint of node) != node.flip
    last_middle_node = null

    while node is not the root {
        parent = node.parent
        is_left_child = parent.children[0] == node

        // Compute where v is in the parent, taking the parent's
        // flip into account.
        is_middle = if is_left_child {
            is_right || (is_middle && !has_right_boundary(node))
        } else {
            is_left || (is_middle && !has_left_boundary(node))
        }
        is_left = (is_left_child != parent.flip) && !is_middle
        is_right = (is_left_child == parent.flip) && !is_middle

        // Go up to the parent
        node = parent

        // If v is in the middle, then it could be the consuming node.
        if is_middle {
            if !has_middle_boundary(node) {
                // This only happens if the vertex is not exposed.
                return node
            }
            last_middle_node = node
        }
    }
    // This only happens when the vertex is exposed.
    return last_middle_node
}
\end{verbatim}
The idea behind \texttt{find\_consuming\_node} is to follow the root path,
keeping track whether the vertex is the left, middle or right boundary vertex of
the current node. If it encounters a node where it should be the middle boundary
vertex, but the node has no middle boundary vertex, then it has found the
smallest cluster containing the vertex without the vertex being a boundary
vertex. If it reaches the root, then the function returns the largest cluster
with the vertex as the central vertex.

The \texttt{is\_left}, \texttt{is\_middle}, and \texttt{is\_right} booleans
store whether the vertex is the left, middle, or right boundary vertex of the
current node respectively. These booleans are ``from the perspective of the
parent'', so if the current node has its flip bit set, then \texttt{is\_left}
and \texttt{is\_right} are swapped.

To see that \texttt{is\_middle} is computed correctly, we note that this holds if
the vertex is the central vertex of the parent. If \texttt{node} is a left child
of the parent (not taking the parent's flip into account), then this happens
when \texttt{v} is the rightmost boundary vertex of \texttt{node}, which happens
if it is to the right, or in the middle when the node has no boundary vertex to
the right. It works analogously when \texttt{node} is a right child. We note
here that \texttt{has\_left\_boundary} must take the flip bit of the given node
into account too.

To see that \texttt{is\_left} is computed correctly, note that a boundary vertex
is to the left when it comes from the left child and isn't the central vertex.
Analogously for \texttt{is\_right}. The expression in the code uses
\texttt{parent.flip} because we are computing the value taking the parent's flip
bit into account.

If we let \(e\) be the edge incident to \(v\) chosen at the start and let \(c\)
be the returned consuming node, then the semi-splay changes the potential by
\(O(1+r(\treeroot(e))-r(e)) - \Omega(\depth(e)) \leq O(\log
n)-\Omega(\depth(c))\), so the amortized cost of
\(\mathtt{find\_consuming\_node}(v)\) is \(O(\log n)-\Omega(\depth(c))\).  Thus,
after a call to \(c=\mathtt{find\_consuming\_node}(v)\) we can do
\(O(\depth(c))\) work ``for free'', which we use e.g.\ in our implementation of
\texttt{expose} and \texttt{deexpose}.

\section{External operations on top trees}\label{sec:dynamic-operations}

This section contains the pseudocode and proofs for the dynamic operations
\texttt{expose}, \texttt{deexpose}, \texttt{link} and \texttt{cut}. 

\subsection{Expose}
\label{sec:expose}

The implementation of \texttt{expose} in this section uses a semi splay followed
by a full splay because that version allows for the simplest implementation, but it is possible to avoid the full splay. How this is done is explained in
Appendix~\ref{sec:expose-semi}.

Pseudocode for an implementation of expose based on \texttt{full\_splay}:

\begin{verbatim}
fn expose(vertex) {
    node = find_consuming_node(vertex) // this contains a semi_splay
    if node is null {
        // The vertex has degree zero.
        vertex.exposed = true
        return null
    }

    while is_path(node) { // rotate_up until consuming node is a point cluster
        parent = node.parent
        push_flip(node)
        node_idx = index of node in parent.children
        rotate_up(node.children[node_idx])
        node = parent
    }

    full_splay(node)

    // Now depth(node)<=1, and node is the consuming point cluster.
    root = null
    while node is not null {
        root = node
        root.num_boundary_vertices += 1
        node = root.parent
    }
    vertex.exposed = true
    return root
}
\end{verbatim}

When a vertex is exposed, it is considered a boundary vertex of all clusters
containing the vertex, so we need to increment the number of boundary vertices
for all clusters it is not already a boundary vertex in. Those clusters are
exactly the consuming node of the vertex and all of its ancestors, but we cannot
just increment the number of boundary vertices for those ancestors, because they
might already have two boundary vertices. Because of this, the operation will
first ``prepare'' the consuming node by ensuring that it has no ancestors with
two boundary vertices, and then actually expose it.

The purpose of the first loop is to ensure that \texttt{node} is a point cluster,
which makes the \texttt{full\_splay} give stronger guarantees about the new
depth. The rotation is allowed by Lemma~\ref{lem:path-rotate}. Additionally, it
can be seen that \texttt{node} is still the consuming node after the loop,
because the rotation makes \texttt{parent} the lowest common ancestor of what
were the two children of \texttt{node} before the rotation.

After the full splay, the actual expose happens. Here, \texttt{node} has
depth at most one by Lemma~\ref{lem:full-splay} since \texttt{node} is still a
point cluster. Additionally, its only boundary vertex cannot be in the middle
because that's where the vertex being consumed is. Therefore, exposing the
vertex does not break the orientation invariant for the root. Finally, since
both \texttt{node} and the root are point clusters, increasing their number of
boundary vertices does not make them invalid.

Each rotation in the \texttt{is\_path(node)} loop changes the potential by
\begin{equation}
  r(\mathtt{node.children}[\mathtt{!node\_idx}] \cup
  \mathtt{sibling}(\mathtt{node}))-r(\mathtt{node})
  \leq r(\mathtt{node.parent}) - r(\mathtt{node}).
\end{equation}
This sum telescopes, so the total change to the
potential by the loop is at most \(\log_2 n\). Additionally, the loop decreases
the depth of \texttt{node} by one each time, so the number of iterations is at
most the depth of the node that was semi-splayed inside
\texttt{find\_consuming\_node}, so by Lemma~\ref{lem:semi-splay-pot}, the code
until just before the full splay has an amortized running time of \(O(\log n)\).
Since the rest of the function also runs in amortized \(O(\log n)\), the total
running time is amortized \(O(\log n)\).


\subsection{Deexpose}

Pseudocode for \texttt{deexpose} can be found below:
\begin{verbatim}
fn deexpose(vertex) {
    root = null
    node = find_consuming_node(vertex)
    while node is not null {
        root = node
        root.num_boundary_vertices -= 1
        node = root.parent
    }
    vertex.exposed = false
    return root
}
\end{verbatim}

When a vertex is exposed, that means that it is a boundary vertex of all
clusters containing it. To deexpose it, we need to remove it as boundary vertex
of the clusters it isn't a true boundary vertex of, which is exactly the
consuming node of the vertex and all of its ancestors. The implementation simply
finds the consuming node, then subtracts one from the number of boundary
vertices from those ancestors.

The algorithm does not break the orientation invariant anywhere because, in each
node we remove the boundary vertex from, it is the other boundary vertex whose
orientation is important, and the other boundary vertex does not move around in
the tree.

The operation does \(O(\depth(c))\) work with \(c\) being the consuming node of
the vertex, but the \texttt{find\_consuming\_node} call changes the potential by
at most \(O(\log n)-\Omega(\depth(c))\), so the amortized running time is
\(O(\log n)\).

\subsection{Link}
Given an implementation of \texttt{expose}, we can implement \texttt{link}:
\begin{verbatim}
fn link(u,v) { // Assumes u and v in trees with no exposed vertices
    Tu = expose(u)
    if Tu is not null and has_left_boundary(Tu) {
        Tu.flip = !Tu.flip
    }
    u.exposed = false
    Tv = expose(v)
    if Tv is not null and has_right_boundary(Tv) {
        Tv.flip = !Tv.flip
    }
    v.exposed = false

    T = new node corresponding to the edge (u,v)
    T.num_boundary_vertices = (Tu is not null)+(Tv is not null)

    if Tu is not null {
        T = new node with children Tu and T
        T.num_boundary_vertices = (Tv is not null)
    }
    if Tv is not null {
        T = new node with children T and Tv
        T.num_boundary_vertices = 0
    }

    return T
}
\end{verbatim}
We can implement \texttt{link} using \texttt{expose} in this manner because an exposed
vertex is always a boundary vertex of all clusters it appears in, which is
exactly what we need to add a new edge to the vertex. The pseudocode puts the
two trees together by first merging the tree containing \(u\) with the new edge,
then merging the result with the tree containing \(v\).

The trees are flipped if necessary to satisfy the orientation invariant in the
new nodes in the tree. Note that the checks for whether to flip require that the
trees do not already have any exposed vertices, as the check would not be able
to tell them apart from the one exposed by \texttt{link}.

The amortized running time of the \texttt{link} operation is \(O(\log n)\)
because that's the amortized cost of the two \texttt{expose} operations, and the
new nodes that it adds to the tree increase the amortization potential by at
most \(2\log_2(n)\).

\subsection{Cut}
Using a full splay and \texttt{deexpose} we can implement \texttt{cut} as follows:
\begin{verbatim}
  fn delete_all_ancestors(node) {
      p = node.parent
      if p is not null {
          s = sibling(node)
          delete_all_ancestors(p)
          s.parent = null
      }
      delete node
  }

  fn cut(e) { // Assumes there are no exposed vertices
      u = e.vertices[0]
      v = e.vertices[1]
      full_splay(e)
      // now depth(e)<=2, and if e is a leaf edge, depth(e)<=1
      delete_all_ancestors(e)
      u.exposed = true
      v.exposed = true
      Tu = deexpose(u)
      Tv = deexpose(v)
      return (Tu,Tv)
  }
\end{verbatim}
Since we assume that there are no exposed vertices, the full splay reduces
the depth to \(2\) when \(e\) is a path cluster, or \(1\) when \(e\) is a point
cluster. For a path cluster, the two top trees that remain
when removing the ancestors of the edge must correspond to the two
trees you get by removing the edge, since otherwise the original top tree would
contain a cluster that isn't connected. For point clusters, removing
the ancestors of the edge does not cut the top tree into several
pieces, so here the remaining top tree also corresponds to the
remaining piece of the underlying tree.

Note that after removing the ancestors, the vertices are marked as boundary
vertices of all clusters they appear in, which is incorrect. Marking them as
exposed restores the top tree invariants, and they can then be deexposed
afterwards.

The amortized running time of the cut operation is \(O(\log n)\) because that's
the amortized cost of a full splay and the \texttt{deexpose} operations.
Removing the edge and its at most \(2\) ancestors only decreases the
amortization potential, so that is also okay.


\section{Implementation and testing}\label{sec:implement}

To help verify the correctness of the pseudocode in this paper, we provide an
implementation of the data structure in the C programming language. The
implementation can be found at: \\
\url{https://gitlab.com/aliceryhl/toptree-c-example}

Whenever a function in the C code also exists as pseudocode in the paper, we
have ensured that they match line-for-line to make it easy to verify that they
implement the algorithm in the same way. There are a few exceptions since the C
code must handle e.g.\ allocation failures, but the equivalence should be clear.

The C code maintains the following user data: For each edge, a weight is stored.
Additionally, for each path cluster, the maximum weight of any edge on the path
between the two boundary vertices is stored. Using this information, we can use
the top tree to dynamically maintain a minimum spanning tree of a weighted graph
as new edges are added to the graph (see e.g.~\cite{HolmLT01}).

To verify the correctness of the C code, we have implemented a testing utility
that generates a random graph with a given number of vertices and edges, then
verifies that the minimum spanning tree generated by our top tree has the same
total weight as the minimum spanning tree generated by running Kruskal's
algorithm~\cite{kruskal1956shortest} on the same graph. Furthermore, our
implementation provides a way to check whether all invariants are satisfied,
which is called periodically during the test. We have run this testing utility
on a larger number of random graphs. The two algorithms agreed on the total
weight every time, and the invariants were satisfied every time we checked them.
This makes it very likely that the implementation is correct.

\pagebreak[4]
\phantomsection\addcontentsline{toc}{section}{References}
\bibliography{toptree}

\newpage
\appendix

\section{Expose that avoids full splay}
\label{sec:expose-semi}

The \texttt{expose} operation in the main text of the paper uses a full splay to move the
consuming node of the vertex very close to the root, which makes the actual
expose rather simple. However, it is also possible to implement expose using
only a semi-splay, reducing the amortized cost of \texttt{expose} (and
\texttt{link}) by a constant factor that may be significant in practice.

This implementation of \texttt{expose} consists of three steps: first semi-splay the
vertex, then ``prepare'' for the expose operation by ensuring that no ancestors
of the consuming node are path clusters, then we actually expose it.

\subsection{Preparing for the expose operation}
\label{sec:prepare-expose}

We define a subroutine that performs some rotations to prepare for the expose
operation. This subroutine ensures that no ancestors of the consuming node are
path clusters. Note that it takes as argument the consuming node of the vertex
rather than the vertex itself. The function returns the new consuming node.

\begin{verbatim}
fn prepare_expose(consuming_node) {
    node = consuming_node
    while node is not the root {
        parent = node.parent
        if is_point(node) {
            node = parent
        } else {
            push_flip(parent)
            push_flip(node)

            sibling = sibling(node)
            sibling_idx = index of sibling in parent.children
            same_side_child = node.children[sibling_idx]
            if is_path(same_side_child) or is_point(sibling) {
                // Case (a),(b),(c),(d)
                other_side_child = node.children[1-sibling_idx]
                rotate_up(other_side_child)
                if node == consuming_node {
                    // Case (a),(b)
                    consuming_node = parent
                }
                node = parent
            } else {
                uncle = sibling(parent)
                gparent = parent.parent
                uncle_idx = index of uncle in gparent.children

                if sibling_idx == uncle_idx {
                    // Case (e)
                    rotate_up(node)
                } else {
                    // Case (f)
                    rotate_up(sibling)
                }
            }
        }
    }
    return consuming_node
}
\end{verbatim}

The correctness of each rotation is seen by considering the figure and noticing
that the new cluster must necessarily be valid in each case.

\begin{figure}[h]
  \begin{subfigure}{0.45\textwidth}
    \begin{tikzpicture}[scale=0.8]
      \node[draw=black,circle] (A) at (0,0) {A};
      \node[draw=black,fill=gray,circle] (B) at (2,1) {B};
      \node[draw=black,circle] (C) at (4,0) {C};
      \node[draw=black,circle] (D) at (-2,1) {D};
      \draw (A) -- (B);
      \draw (B) -- (C);
      \draw (A) -- (D);
      \draw[dashed] (2,0.1) ellipse (2.5cm and 1.4cm);
    \end{tikzpicture}
    \caption{B is exposed, next edge on A. We merge AB and DA.}
  \end{subfigure}
  \hfill
  \begin{subfigure}{0.45\textwidth}
    \begin{tikzpicture}[scale=0.8]
      \node[draw=black,circle] (A) at (0,0) {A};
      \node[draw=black,fill=gray,circle] (B) at (2,1) {B};
      \node[draw=black,circle] (C) at (4,0) {C};
      \node[draw=black,circle] (D) at (6,1) {D};
      \draw (A) -- (B);
      \draw (B) -- (C);
      \draw (C) -- (D);
      \draw[dashed] (2,0.1) ellipse (2.5cm and 1.4cm);
    \end{tikzpicture}
    \caption{B is exposed, next edge on C. We merge BC and CD. Symmetric  to Case (a).}
  \end{subfigure}

  \begin{subfigure}{0.45\textwidth}
    \begin{tikzpicture}[scale=0.8]
      \node[draw=black,fill=gray,circle] (A) at (0,0) {A};
      \node[draw=black,circle] (B) at (2,1) {B};
      \node[draw=black,circle] (C) at (4,0) {C};
      \node[draw=black,circle,dashed] (D) at (4,2) {D};
      \draw (A) -- (B);
      \draw (B) -- (C);
      \draw (B) -- (D);
      \draw[dashed] (2,0.1) ellipse (2.5cm and 1.4cm);
    \end{tikzpicture}
    \caption{A is exposed, edge to D is point cluster and connects on B. We
    merge AB and BD.}
  \end{subfigure}
  \hfill
  \begin{subfigure}{0.45\textwidth}
    \begin{tikzpicture}[scale=0.8]
      \node[draw=black,fill=gray,circle] (A) at (0,0) {A};
      \node[draw=black,circle] (B) at (2,1) {B};
      \node[draw=black,circle] (C) at (4,0) {C};
      \node[draw=black,circle] (D) at (6,1) {D};
      \draw (A) -- (B);
      \draw (B) -- (C);
      \draw (C) -- (D);
      \draw[dashed] (2,0.1) ellipse (2.5cm and 1.4cm);
    \end{tikzpicture}
    \caption{A is exposed, next edge connects on C. We merge BC and CD.}
  \end{subfigure}

  \begin{subfigure}{0.45\textwidth}
    \begin{tikzpicture}[scale=0.8]
      \node[draw=black,fill=gray,circle] (A) at (0,0) {A};
      \node[draw=black,circle] (B) at (2,1) {B};
      \node[draw=black,circle] (C) at (4,0) {C};
      \node[draw=black,circle] (D) at (4,2) {D};
      \node[draw=black,circle] (E) at (6,1) {E};
      \draw (A) -- (B);
      \draw (B) -- (C);
      \draw (B) -- (D);
      \draw (D) -- (E);
      \draw[dashed] (2,0.1) ellipse (2.5cm and 1.4cm);
    \end{tikzpicture}
    \caption{A is exposed, edge to D is path cluster on B, and edge to E
    connects to B. 
    We merge BD and DE.}
  \end{subfigure}
  \hfill
  \begin{subfigure}{0.45\textwidth}
    \begin{tikzpicture}[scale=0.8]
      \node[draw=black,fill=gray,circle] (A) at (0,0) {A};
      \node[draw=black,circle] (B) at (2,1) {B};
      \node[draw=black,circle] (C) at (4,0) {C};
      \node[draw=black,circle] (D) at (4,2) {D};
      \node[draw=black,circle] (E) at (6,1) {E};
      \draw (A) -- (B);
      \draw (B) -- (C);
      \draw (B) -- (D);
      \draw (C) -- (E);
      \draw[dashed] (2,0.1) ellipse (2.5cm and 1.4cm);
    \end{tikzpicture}
    \caption{A is exposed, edge to D is path cluster on B, and edge to E
    connects to C. We swap BD and CE, leaving situation (d) for next iteration.}
    \label{fig:prepare-expose-cases-swap}
  \end{subfigure}
  \caption{The different situations one can encounter during the subroutine from
  section~\ref{sec:prepare-expose} when the current node has two boundary
  vertices.  The two edges AB and BC inside the dashed circles are the two
  children of the current node. The gray vertex corresponds to the vertex we
  want to expose. The two other vertices inside the dashed circle are boundary
  vertices of the current node. The edge connecting to D is the sibling of the
  current node. The edge connecting to E is the sibling of the parent of the
  current node.} 
  \label{fig:prepare-expose-cases}
\end{figure}
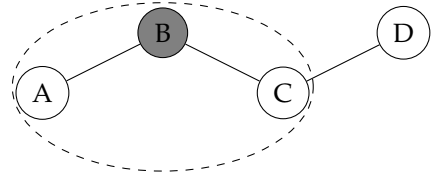
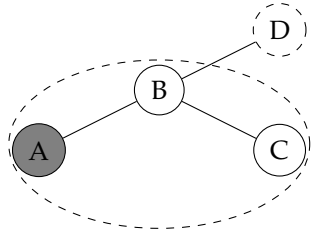
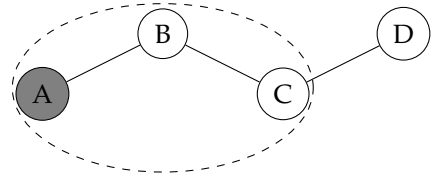
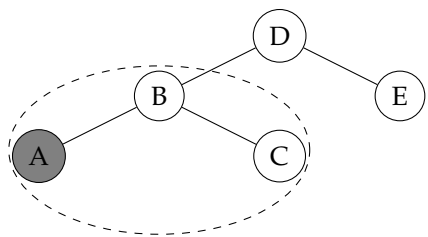
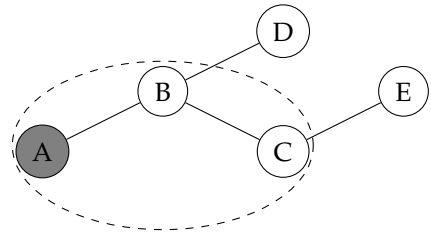

\begin{lemma}\label{lem:prepare-expose-prepares}
  When the \texttt{prepare\_expose} function returns, every cluster containing
  the vertex being exposed is either a point cluster, or already has the vertex
  as a boundary vertex.
\end{lemma}
\begin{proof}
  The above algorithm will maintain the following invariant: At every iteration,
  every node on the root path from the consuming node of the vertex and until
  the current node (exclusive) is a point cluster. Initially, this holds
  vacuously since the current node is equal to the consuming node.

  Whenever the algorithm reaches a point cluster, it will change the current
  node to its parent, which does not break the invariant. When it reaches a path
  cluster, it will be in one of the six cases in
  Figure~\ref{fig:prepare-expose-cases}. Cases (a), (b), (c) and (d) are handled
  in the \texttt{is\_path(same\_side\_child) or is\_point(sibling)} branch. Case
  (e) is handled in the \texttt{sibling\_idx == uncle\_idx} branch and case
  (f) is handled in the \texttt{else} branch. In case (a) and (b), the current
  node is equal to the consuming node both before and after the operation, so
  the invariant holds vacuously. In case (c), the cluster \(\{\text{AB},
  \text{BD}\}\) is added below the node, but it is a point cluster since only B
  is a boundary vertex, so this maintains the invariant. In cases (d), (e) and
  (f), no clusters relevant to the invariant are modified, so the invariant is
  maintained.

  The algorithm terminates when it reaches the root. Here, the invariant implies
  the lemma because every cluster containing the vertex is either the root, a
  point cluster according to the invariant, or a proper descendant of the
  consuming node. The root must be a point cluster too, since it is a
  precondition of \texttt{expose} that at most one other vertex is already
  exposed. The proper descendants of the consuming node that contain the vertex
  must already have the vertex as a boundary vertex by definition of consuming
  node.
\end{proof}

\begin{lemma}
  The node returned by \texttt{prepare\_expose} is the new consuming node of the
  vertex after the changes.
\end{lemma}
\begin{proof}
  Operations (a) and (b) are the only ones that change the consuming node of the
  vertex, and the current node is the consuming node exactly those cases.
  Therefore, the \texttt{node == consuming\_node} condition in the pseudocode
  updates the consuming node in exactly the cases where it changes.
\end{proof}

\begin{lemma}
  The \texttt{prepare\_expose} function does \(O(\depth(\mathtt{node}))\) work
  and increases the amortization potential by at most \(O(\log n)\).
\end{lemma}
\begin{proof}
  Since Figure~\ref{fig:prepare-expose-cases-swap} is always followed by case
  (d), we treat an (f) followed by (d) as a single iteration in this proof.
  Since the lemma is formulated in terms of big-O, this doesn't change anything.

  Every iteration of \texttt{prepare\_expose} decreases the depth of
  \texttt{node} by one, so it makes \(O(\depth(\mathtt{node}))\) iterations of
  constant cost as desired.

  Each operation changes the amortization potential by:
  \begin{description}
    \item[(a)] \(r(\{\text{AB},\text{AD}\})-r(\{\text{AB},\text{BC}\}) \leq
      r(\{\text{AB},\text{AD},\text{BC}\})-r(\{\text{AB},\text{BC}\}) =
      r(\text{new node})-r(\text{old node})\)
    \item[(b)] \(r(\{\text{BC},\text{CD}\})-r(\{\text{AB},\text{BC}\}) \leq
      r(\{\text{AB},\text{BC},\text{CD}\})-r(\{\text{AB},\text{BC}\}) =
      r(\text{new node})-r(\text{old node})\)
    \item[(c)] \(r(\{\text{AB},\text{BD}\})-r(\{\text{AB},\text{BC}\}) \leq
      r(\{\text{AB},\text{BD},\text{BC}\})-r(\{\text{AB},\text{BC}\}) =
      r(\text{new node})-r(\text{old node})\)
    \item[(d)] \(r(\{\text{BC},\text{CD}\})-r(\{\text{AB},\text{BC}\}) \leq
      r(\{\text{AB},\text{BC},\text{CD}\})-r(\{\text{AB},\text{BC}\}) =
      r(\text{new node})-r(\text{old node})\)
    \item[(e)] \(r(\{\text{BD},\text{DE}\})-r(\{\text{BD},\mathtt{node}\}) \leq
      r(\{\text{BD},\text{DE},\mathtt{node}\})-r(\{\text{BD},\mathtt{node}\}) =
      r(\text{parent}(\text{new node})) - r(\text{parent}(\text{old node}))\)
    \item[(f)+(d)] \(r(\{\text{BC},\text{CE}\})+r(\{\text{AB},\text{BC},\text{CE}\})
      -r(\{\text{AB},\text{BC}\})-r(\{\text{AB},\text{BC},\text{BD}\}) \leq
      r(\{\text{AB},\text{BC},\text{CE},\text{BD}\})+r(\{\text{AB},\text{BC},\text{CE}\})
      -r(\{\text{AB},\text{BC}\})-r(\{\text{AB},\text{BC},\text{BD}\}) =
      r(\text{parent(new node)})+r(\text{new node})-r(\text{old node})-r(\text{parent(old node)})\)
  \end{description}
  The total change in potential is the sum of the above expressions, and the sum
  of the first \(i\) terms will always be at most \(r(\text{node in $i$th
  iteration}) + r(\text{parent(node) in $i$th iteration})\) for every \(i\), so
  the total change in the potential is at most \(2\log_2 n\) as desired.
\end{proof}

\subsection{Exposing a prepared vertex}

Once the vertex has been prepared, we need to make it an actual boundary vertex
of the consuming node and its ancestors. This amounts to going through each node
on the root path from the consuming node and increasing its number of boundary
vertices by one. It is also sometimes necessary to flip the nodes, since if the
node is a left child and its rightmost boundary vertex is the middle vertex,
then if the vertex being exposed becomes the right boundary vertex of the node,
then the middle vertex will no longer be the rightmost boundary vertex. In this
case, we should flip the node so that the new boundary vertex becomes the left
boundary vertex instead.

The pseudocode for this is as follows:
\begin{verbatim}
fn expose_prepared(consuming_node) {
    from_left = false
    from_right = false
    node = consuming_node
    while true {
        node.num_boundary_vertices += 1
        parent = node.parent
        if parent is null {
            // Here we return the root node.
            return node
        }

        is_left_child_of_parent = parent.children[0] == node;
        is_right_child_of_parent = !is_left_child_of_parent;

        if (is_left_child_of_parent && from_right)
            || (is_right_child_of_parent && from_left)
        {
            node.flip = !node.flip
        }

        from_left = is_left_child_of_parent != parent.flip
        from_right = is_right_child_of_parent != parent.flip
        node = parent
    }
}
\end{verbatim}
In each iteration, increasing the number of boundary vertices will introduce a
new boundary vertex to the current node. This new boundary vertex could be to
the left, in the middle, or to the right. If our vertex is in the middle, then the
existing boundary vertex isn't in the middle so a flip is not necessary. If we
are to the left/right, then we should flip if we become the innermost boundary
vertex, which the code checks by looking at whether we are the left or right
child of the parent, ignoring the parent's flip since the parent's flip is only
relevant when comparing the parent to the grandparent. We note that the first
iteration is the only time that the new boundary vertex can be in the middle,
which is represented as \(\mathtt{from\_left} = \mathtt{from\_right} =
\mathtt{false}\). If the vertex being exposed is not the middle vertex of the
consuming node, then the consuming node is a leaf node. However, since leaf
nodes never have a middle boundary vertex, we also don't need to flip in this
situation.

The \texttt{expose\_prepared} function does
\(O(\depth(\mathtt{consuming\_node}))\) work and does not change the
amortization potential.

\subsection{The full expose operation}
Using the two pieces introduced above, the full expose operation is the
following:
\begin{verbatim}
fn expose(vertex) {
    consuming_node = get_consuming_node(vertex)
    if consuming_node is not null {
        consuming_node = prepare_expose(consuming_node)
        root = expose_prepared(consuming_node)
        vertex.exposed = true
        return root
    } else {
        // The vertex has degree zero.
        vertex.exposed = true
        return null
    }
}
\end{verbatim}
Adding up the contributions of each call, the above function does
\(O(\depth(x))\) work and changes the potential by at most \(\Delta\Phi \leq
O(\log n) - \Omega(\depth(x))\), where \(x\) is the edge incident to the vertex
that \texttt{get\_consuming\_node} chooses. Thus, the amortized running time of
\texttt{expose} is \(O(\log n)\).

%
%
%
%
%
%
%
%
%
%
%
%
%

\end{document}